\documentclass[12pt]{article}
\usepackage{booktabs}
\usepackage{multirow}
\usepackage{mathrsfs}
\usepackage{amsthm}
\usepackage{fullpage}
\usepackage{footnote}
\usepackage{tabu}
\usepackage[title]{appendix}
\usepackage{caption}
\usepackage{amsmath}
\usepackage{setspace}
\usepackage{stmaryrd}
\providecommand{\keywords}[1]
{
  \small
  \textbf{\textit{Keywords---}} #1
}
\usepackage{etoolbox}

\makeatletter
\patchcmd{\@makecaption}
  {\parbox}
  {\advance\@tempdima-\fontdimen2}
  {}{}
\makeatother

\usepackage{titling}
\usepackage{amsmath,bbm,amssymb, amsfonts}
\usepackage{mathtools}
\usepackage{graphicx,psfrag,epsf}
\usepackage{enumitem}
\usepackage{natbib}
\usepackage{url} 
\usepackage{color}
\usepackage{algorithmicx}
\usepackage{booktabs}
\usepackage{siunitx}
\usepackage{array}
\usepackage{rotating}

\usepackage{adjustbox}

\usepackage{algorithm}
\usepackage{algpseudocode}

\usepackage{tikz}

\usepackage{url}

\usepackage{multicol}
\usepackage{xcolor}
\usepackage{algorithm}

\renewcommand{\Pr}{{\bf Pr}}

\usepackage[colorlinks, linkcolor=blue, anchorcolor=blue, citecolor=blue, hypertexnames=false]{hyperref}

\newtheorem{theorem}{Theorem}
\newtheorem{assumption}{Assumption}
\newtheorem{thm}{Theorem}[section]
\newtheorem{lem}[thm]{Lemma}

\usepackage[left=1in,right=1in,top=1in,bottom=1in]{geometry}

\def\spacingset#1{\renewcommand{\baselinestretch}%
{#1}\small\normalsize}

\makeatletter
\newcommand{\SuppLabels}{\gdef\newlabel##1{\@newl@bel r{supp:##1}}}
\AtEndDocument{\gdef\newlabel{\@newl@bel r}}
\newcommand{\StartSupplement}{%
  \immediate\write\@auxout{\string\SuppLabels}%
  \NewCommandCopy{\MainRef}{\ref}%
  \protected\def\ref##1{\MainRef{supp:##1}}%
  \setcounter{section}{0}%
  \setcounter{equation}{0}%
  \setcounter{figure}{0}%
  \setcounter{table}{0}%
  \setcounter{algorithm}{0}%
  \renewcommand{\theequation}{S\arabic{equation}}%
  \renewcommand{\thefigure}{S\arabic{figure}}%
  \renewcommand{\thetable}{S\arabic{table}}%
  \let\assumption\relax
  \let\endassumption\relax}
\makeatother

\title{\bf Bayesian Multiple Imputation for Functional Data}
\author{Ziren Jiang, Lei Xuan, Eric F. Lock, Erjia Cui\thanks{corresponding author: ecui@umn.edu}\\
	Division of Biostatistics and Health Data Science,\\
	School of Public Health, University of Minnesota}
\date{}

\begin{document}

\maketitle

\begin{abstract}
Missing data are pervasive in modern functional datasets, where trajectories are often sparsely or irregularly observed. Although Functional Principal Component Analysis (FPCA) is widely used to reconstruct incomplete curves, existing approaches typically employ single imputation, leading to overly optimistic inferences in downstream analyses.  To address these challenges, we develop a novel Bayesian multiple imputation framework for functional data (BAMIFun).
For single-level functional data, we impose a Bayesian low-rank model that incorporates penalized spline representations to enforce smoothness of the functional domain and derive an efficient Gibbs sampler algorithm for posterior computation.
In addition, we demonstrate and validate how to properly account for estimation uncertainties in downstream analysis.
Furthermore, we extend the framework to multiway functional data using Functional Tensor Singular Value Decomposition (FTSVD) model, enabling Bayesian multiple imputation in settings not supported by existing methods. Simulation studies show that BAMIFun achieves substantially improved coverage and more reliable downstream inference compared to existing methods, while maintaining similar imputation accuracy. Case studies using a physical activity dataset and an infant gut microbiome dataset further demonstrate the practical advantages of our proposed methods under severe missingness.

\end{abstract}

\newpage
\setstretch{1.775}
\setlength{\abovedisplayskip}{7pt}%
\setlength{\belowdisplayskip}{7pt}%
\setlength{\abovedisplayshortskip}{5pt}%
\setlength{\belowdisplayshortskip}{5pt}%

\section{Introduction}
Functional data are often observed on sparse or irregular grids.
A wide range of methods has been developed for modeling such data, including Functional Principal Component Analysis (FPCA) \citep{yao2005functional}, Bayesian joint registration and curve estimation \citep{matuk2022bayesian}, and matrix completion \citep{kidzinski2024modeling}. 
Despite these advances, statistical analysis of sparsely observed functional data remains challenging. Major difficulties include the implicit reliance of many methods on dense design assumptions \citep{kong2016classical}, the limited number of tailored software programs 
\citep{wood2015package}, 
and performance degradation under severe missingness. 


A common approach to handling sparse or irregular functional data is to impute missing observations and reconstruct subject-specific trajectories on a regular grid, a method first proposed by \citet{yao2005functional} using FPCA. Specifically, each trajectory is reconstructed by estimating principal component scores via conditional expectations. 
This FPCA-based approach has since been extended to multilevel \citep{di2009multilevel, zipunnikov2011multilevel, cui2023fast}, longitudinal \citep{greven2011longitudinal}, structural \citep{shou2015structured}, and multivariate \citep{chiou2014multivariate, happ2018multivariate} settings. In what follows, we refer to this type of method as the ``PACE approach'', as coined in \citet{yao2005functional}. The PACE approach produces a single imputed value at each missing location, 
hence belonging to the ``single imputation" approaches. Because single imputation treats the reconstructed trajectories as the truth, the uncertainty measures in downstream analyses can often be overly optimistic, particularly in cases of severe sparsity \citep{rao2021modern, petrovich2022highly}.

In contrast, the multiple imputation framework \citep{rubin1996multiple} explicitly incorporates the imputation uncertainty by generating multiple plausible completions of missing observations. The variability between these completed datasets reflects the uncertainty inherent in the imputation process, leading to more valid inference for downstream analysis. Several multiple imputation approaches have been proposed for sparse functional data. \citet{petrovich2022highly} introduced a frequentist multiple imputation procedure that leverages information from a scalar outcome, assuming a generalized additive model 
and drawing imputations from the corresponding conditional distribution. \citet{rao2021modern} adapted the missForest algorithm to accommodate functional variables. \citet{he2011functional} modeled functional and scalar covariates jointly through a functional mixed-effects framework and developed a Bayesian multiple imputation strategy. \citet{jang2021bayesian} proposed a Bayesian multiple imputation method for bivariate functional data that exploits the correlation structure between two functional covariates. 

Although effective in their respective contexts, existing methodologies often require auxiliary information, such as scalar outcomes, scalar covariates, or supplementary functional covariates, for imputation. When such information is unavailable or when its relationship with functional data is misspecified, the imputation performance becomes questionable. Moreover, we are not aware of any existing multiple imputation method for more complex functional data structures, such as multiway functional data. Multiway data
consist of observations indexed over more than two modes (e.g., subject $\times$ feature $\times$ visit) and arise naturally in modern applications such as 
genomics \citep{li2025integrative}. When one of these modes is a continuum such as time or space, the data become multiway functional. 
Recently, \citet{jiang2025bamita} proposed a Bayesian multiple imputation algorithm for tensor arrays (BAMITA), which imputes tensor data via a low-rank CP decomposition. However, BAMITA is designed for discrete tensors: it treats the temporal argument as an unordered mode and does not exploit the smoothness of functional data. As a result, it cannot borrow information from neighboring points for imputation. 


To fill in these critical methodological gaps, we develop BAMIFun: a Bayesian multiple imputation framework for functional data. The BAMIFun algorithm (1) embeds penalized splines within a Bayesian framework to enforce trajectory smoothness; (2) performs multiple imputation for single-level functional data via a low-rank representation and achieves valid downstream inference by pooling estimates using Rubin’s rules; and (3) provides the first multiple imputation procedure for multiway functional data by developing a Bayesian sampler for the functional tensor singular value decomposition (FTSVD) model \citep{han2024guaranteed}.
Through extensive simulations and real-data applications, we demonstrate that BAMIFun achieves similar imputation accuracy to the PACE approach, while providing substantially improved coverage for imputed entries and more reliable uncertainty quantification in downstream functional regressions. Moreover, BAMIFun achieves substantially higher imputation accuracy compared to BAMITA, especially in cases of severe missingness.


The remainder of the manuscript is organized as follows. Section \ref{sec:2} introduces the problem setup. Section \ref{sec:3} presents BAMIFun for single-level functional data. Section \ref{sec:4} extends the framework to multiway functional data. Sections \ref{sec:5} and \ref{sec:6} evaluate BAMIFun using simulated and real-world data, respectively. Section \ref{sec:7} concludes with a discussion.

\section{Notation and Setup}\label{sec:2}
We denote the single-level functional data as $\{X_i(t_k)\}$, where $i = 1,\ldots,N$ indexes subjects and $\{t_k \in \mathcal{I}$, $k = 1,\ldots,K\}$ represents a union of observation points over a continuum $\mathcal{I}$.
We consider the matrix representation for the functional data $\mathbf{X}\in \mathbb{R}^{N\times K}$, where the entry in the $i$-th row and $k$-th column of $\mathbf{X}$ is $\mathbf{X}_{ik} =X_i(t_k)$. In practice, not all $X_i(t_k)$ are observed for $k=1,...,K$. Let $\mathcal{O}=\{\mathcal{O}_{ik}\}_{i=1,...,N;k=1,...,K} \in \{0,1\}^{N\times K}$ be the observation indicator matrix, where $\mathcal{O}_{ik}=1$ if $X_i(t_k)$ is observed for subject $i$ at time point $t_k$ and $\mathcal{O}_{ik}=0$ otherwise. Let $Y_i$ denote a scalar outcome for subject $i$, which is used exclusively in downstream analyses and is not incorporated into the imputation process. 

For multiway functional data, denote the data as $\{X_{ij}(t_k)\}$, where $i=1,...,N$ represents each subject, $j=1,...,J$ indexes another mode such as visits, and $k=1,...,K$ indexes the observation points. 
We use a tensor $\mathcal{X}$ with the dimensionality of $N\times J\times K$ to represent the multiway functional data, where $\mathcal{X}_{ijk}=X_{ij}(t_k)$. We denote $\mathcal{O}=\{\mathcal{O}_{i j k}\}_{i=1,...,N; j=1,...,J; k=1,...,K}$ as the observation indicator tensor. 
Let $\mathbf{a} \circ\mathbf{b}$ denote the outer product of two vectors $\mathbf{a}$ and $\mathbf{b}$. Let $\mathbf{A} \otimes \mathbf{B}$ denote the Kronecker product of matrices $\mathbf{A}$ and $\mathbf{B}$, and $\mathbf{A} \odot \mathbf{B}$ denote the \it{Khatri-Rao} \normalfont product.
Let $\textnormal{Vec}(\cdot)$ denote the vectorization operator of matrix or tensor object. For matrix $\mathbf{A}$, let $\mathbf{A}^T$ denote its transpose, $\mathbf{A}_{ik}$ denote the element in the $i$-th row and $k$-th column, $\mathbf{A}_{i\cdot}$ denote its $i$-th row, and $\mathbf{A}_{\cdot k}$ denote its $k$-th column.

\vspace{-4ex}
\section{Single-Level Functional Data}\label{sec:3}
\subsection{The Model}

We assume that the functional data follow the structure $X_i(t) = \mu(t)+\Tilde{X}_i(t)+\epsilon_i(t)$, where $\mu(t) = \mathbb{E}[X_i(t)]$ for any fixed $t$, $\Tilde{X}_i(t)$ is a stochastic process with mean zero, and $\epsilon_i(t)$ is a white noise. Without loss of generality, we assume $\mu(t)=0$ (i.e., the data are centered) for the remainder of the manuscript. 
Our Bayesian multiple imputation framework assumes that $\mathbf{X}$ has the following low-rank structure with rank $R$:
\begin{equation}\label{eq: matrix_model_iidGaussian}
    \mathbf{X} = \mathbf{V}\mathbf{U}^T + \mathbf{E}\;,
\end{equation}
where $\mathbf{V}$ is an $N\times R$ matrix, $\mathbf{U}$ is a $K\times R$ matrix, and $\mathbf{E}$ is an $N\times K$ matrix of i.i.d. Gaussian error terms with mean 0 and variance $\sigma^2$. 
The low-rank structure presented in \eqref{eq: matrix_model_iidGaussian} is common for functional data. For example, under the FPCA decomposition, $R$ is the number of principal components, $\mathbf{U}$ is the matrix of eigenfunctions, and $\mathbf{V}$ is the matrix of corresponding scores. Note that model \eqref{eq: matrix_model_iidGaussian} identifies $\mathbf{V}$ and $\mathbf{U}$ only up to the transformation $(\mathbf{V}\mathbf{Q},\mathbf{U}\mathbf{Q}^{-T})$ for any invertible $R\times R$ matrix $\mathbf{Q}$, so the individual factors are not separately identifiable. Because our imputation depends on the factors only through the transformation-invariant reconstruction $\mathbf{V}\mathbf{U}^{T}$, the imputed values and their posterior credible intervals remain well-defined.



The functional data $\mathbf{X}$ is assumed to be smooth along the functional domain $\mathcal{I}$ by the low-rank factor matrix $\mathbf{U}$.
Specifically, we further represent $\mathbf{U}^T$ using $L$ spline basis functions such that
$\mathbf{U}^T = \mathbf{B} \mathbf{\Theta}$,
where $\mathbf{\Theta}$ is an $L\times K$ matrix of spline basis functions evaluated over the grid $\{t_k\}_{k=1}^K$, and $\mathbf{B}$ is the corresponding $R\times L$ matrix of spline coefficients. Model \eqref{eq: matrix_model_iidGaussian} can then be re-expressed as $\mathbf{X} = \mathbf{V} \mathbf{B} \mathbf{\Theta} + \mathbf{E}$. The additional smoothness constraint on $\mathbf{U}$ can be induced by adding a penalty on the integral of the squared second derivative $\int (u_r''(t))^2dt$, where $u_r''(t)$ is the second-order derivative of the $r$-th column $u_r(t)$ of the matrix $\mathbf{U}$. 
Under the spline basis expansion $\mathbf{U}^T = \mathbf{B} \mathbf{\Theta}$, we have
\begin{equation}\label{eq: penalized_spline}
    \begin{split}
        \int (u_r''(t))^2dt & = \int \bigg(\sum_{l=1}^L \mathbf{B}_{rl}\mathbf{\Theta}_{l\cdot}''(t)\bigg)^2 dt
        = \int \sum_{l_1=1}^L\sum_{l_2=1}^L \mathbf{B}_{rl_1}\mathbf{\Theta}_{l_1\cdot}''(t)\,\mathbf{\Theta}_{l_2\cdot}''(t)\,\mathbf{B}_{rl_2} dt\\
        & = \mathbf{B}_{r\cdot}^T \bigg(\int \mathbf{\Theta}''(t)[\mathbf{\Theta}''(t)]^Tdt\bigg)\mathbf{B}_{r\cdot} = \mathbf{B}_{r\cdot}^T \mathbf{P} \mathbf{B}_{r\cdot}\;,
    \end{split}
\end{equation}
where $\mathbf{B}_{r\cdot}$ is the $r$-th row of matrix $\mathbf{B}$, $\mathbf{\Theta}_{l\cdot}(t)$ is the $l$-th spline basis, and $\mathbf{P}$ is the $L\times L$ penalty matrix with the element of $l_1$-th row and $l_2$-th column being $\mathbf{P}_{l_1l_2} = \int [\mathbf{\Theta}_{l_1\cdot}''(t)]^T[\mathbf{\Theta}_{l_2\cdot}''(t)]dt$. As shown in \cite{jiang2025tutorial}, incorporating the penalty $\int (u_r''(t))^2dt$ into the log-likelihood is equivalent to imposing a multivariate normal prior $p(\mathbf{B}_{r\cdot}) \propto \exp \bigg(-\frac{\mathbf{B}_{r\cdot}^T \mathbf{P} \mathbf{B}_{r\cdot}}{2\sigma_B^2}\bigg)$
on each row of the spline coefficients matrix $\mathbf{B}$, where $\sigma_B^2$ is the smoothing parameter. 

For ease of notation, 
we vectorize $\mathbf{X} = \mathbf{V} \mathbf{B} \mathbf{\Theta} + \mathbf{E}$, which gives
\begin{equation}\label{eq: basis_expansion_model}
    \begin{split}
        \textnormal{Vec}(\mathbf{X}) &=  \textnormal{Vec}(\mathbf{V} \mathbf{B} \mathbf{\Theta}) +  \textnormal{Vec}(\mathbf{E}) = (\mathbf{\Theta}^T \otimes \mathbf{V})\textnormal{Vec}(\mathbf{B})+  \textnormal{Vec}(\mathbf{E})\;.
    \end{split}
\end{equation}
Given the vectorized model \eqref{eq: basis_expansion_model}, the prior for $\textnormal{Vec}(\mathbf{B})$ becomes
\begin{equation}\label{eq:priorforB}
    p(\textnormal{Vec}(\mathbf{B})\mid\sigma_B^2)\propto (\sigma_B^2)^{-qR/2}\exp\left(- \frac{1}{2\sigma_B^2}\textnormal{Vec}(\mathbf{B})^T(\mathbf{P}\otimes\mathbf{I}_R)\textnormal{Vec}(\mathbf{B})\right)\;,
\end{equation}
where $\mathbf{I}_R$ is the $R\times R$ identity matrix and $q=\textnormal{rank}(\mathbf{P})=L-2$ for the second-derivative penalty. { Because $\mathbf{P}$ is rank deficient, this prior is partially improper---it is flat on the null space of $\mathbf{P}\otimes\mathbf{I}_R$. In the Supplementary Material Section 1.3, we present a reparametrization to decompose the spline coefficients into unpenalized fixed effects and penalized random effects. For the remaining parameters we adopt proper priors $\mathbf{V}_{ir}\stackrel{\textnormal{iid}}{\sim} N(0,\sigma_V^2)$, $\sigma^2\sim\textnormal{IG}(a_\sigma,b_\sigma)$, and $\sigma_B^2\sim\textnormal{IG}(a_B,b_B)$, with $\sigma_V^2=10^6$ and $a_\sigma=b_\sigma=a_B=b_B=0.01$ by default.}


The final Bayesian model for functional data imputation becomes
\begin{equation}\label{eq:single.level.model}
    \begin{cases}
      \qquad \textnormal{Vec}(\mathbf{X}) = (\mathbf{\Theta}^T \otimes \mathbf{V})\textnormal{Vec}(\mathbf{B})+  \textnormal{Vec}(\mathbf{E})\;;\\
      \qquad \textnormal{Vec}(\mathbf{E})\sim N(\boldsymbol{0}, \sigma^2\mathbf{I}_{NK})\;; \\
      \qquad \textnormal{Vec}(\mathbf{B}) \sim p(\textnormal{Vec}(\mathbf{B})),\; \;\; \mathbf{V} \sim p(\mathbf{V}),\;\;\; \sigma^2\sim p(\sigma^2),\;\textnormal{and} \;\;\; \sigma_B^2\sim p(\sigma_B^2)\;.
    \end{cases}       
\end{equation}

Here $p(\cdot)$ represents the prior for each of the parameters. For model \eqref{eq:single.level.model}, we treat $\sigma_B^2$ as an unknown parameter and obtain posterior samples accordingly. 
The full conditional posterior distributions under this prior specification are derived in the Supplementary Material Section 1. 
We also present an alternative approach that selects $\sigma_B^2$ via cross-validation (CV) { in the Supplementary Material Section 4}. Given the substantially higher computational cost with CV, we recommend estimating $\sigma_B^2$ within the likelihood-based Bayesian model \eqref{eq:single.level.model}.

Although the prior on $\mathbf{B}$ is partially improper, our Theorem \ref{the1} below establishes that the resulting posterior distribution is proper under a mild, checkable condition.

\vspace{-2ex}
\begin{assumption}\label{ass:supp-pattern}
There exists a subset $\mathcal{S}\subseteq\{1,\dots,N\}$ with
$|\mathcal{S}|\geq R+2$ such that every subject $i\in\mathcal{S}$ is
observed at two or more distinct time points; that is,
$\#\{t_k:\mathcal{O}_{ik}=1\}\geq 2$.
\end{assumption}

\vspace{-2ex}
Assumption~\ref{ass:supp-pattern} is weak and readily checkable: at least $R+2$
subjects must have two or more observations, and these time points need not be common across subjects. Under this condition, Theorem \ref{the1} establishes the propriety of the posterior distribution in model \eqref{eq:single.level.model}. We defer the proof to the Supplementary Material Section 2. 

\vspace{-2ex}
\begin{theorem}\label{the1}
Consider model \eqref{eq:single.level.model} with aforementioned
priors, where
$a_\sigma,b_\sigma,a_B,b_B>0$ and $\sigma_V^2>0$.  Under
Assumption~\ref{ass:supp-pattern}, the posterior distribution of
$(\mathbf{V},\mathbf{B},\sigma^2,\sigma_B^2)$ given the observed data
$\mathbf{X}_{\mathrm{o}}=\{\mathbf{X}_{ik}:\mathcal{O}_{ik}=1\}$ is proper, i.e.,
the observed-data marginal likelihood satisfies $0<m(\mathbf{X}_{\mathrm{o}})<\infty$.
\end{theorem}

\vspace{-3ex}
\subsection{Gibbs Sampler with Full Data}\label{sec:3.2}
We now provide full conditionals for the Gibbs sampling procedure given the fully observed functional data $\{X_i(t_k)\}_{i=1,...,N;k=1,...,K}$. 
For model \eqref{eq: matrix_model_iidGaussian} (and equivalently \eqref{eq: basis_expansion_model}), the posterior distributions used in the Gibbs sampler, conditional on the full data and rank $R$, are given below.  For detailed derivations, see the Supplementary Material Section 1. 

\begin{itemize}
    \item Given $\mathbf{U}^T=\mathbf{B}\mathbf{\Theta}$, $\sigma^2$, and $\sigma_B^2$, for $i=1,\ldots,N$, draw the $i$-th row $\boldsymbol{v}_{i\cdot}$ of $\mathbf{V}$ from
    \begin{equation*}
       \boldsymbol{v}_{i\cdot}^T \sim N\!\left(\tfrac{1}{\sigma^2}\boldsymbol{\Sigma}_v\mathbf{U}^T\mathbf{X}_{i\cdot}^T,\;\; \boldsymbol{\Sigma}_v\right)\;,\qquad \textnormal{where} \;\; \boldsymbol{\Sigma}_v=\Bigl(\tfrac{1}{\sigma^2}\mathbf{U}^T\mathbf{U}+\tfrac{1}{\sigma_V^2}\mathbf{I}_R\Bigr)^{-1}\;.
    \end{equation*}
    
    \item Given $\mathbf{V}$, $\sigma^2$, and $\sigma_B^2$, draw $\textnormal{Vec}(\mathbf{B})$ following $\textnormal{Vec}(\mathbf{B}) \sim N(\boldsymbol{\mu}_B, \mathbf{\Sigma}_{B})$
    where
    \begin{equation*}
        \mathbf{\Sigma}_B = \bigg(\frac{1}{\sigma^2}(\mathbf{\Theta}^T \otimes \mathbf{V})^T(\mathbf{\Theta}^T \otimes \mathbf{V})+\frac{1}{\sigma_B^2}(\mathbf{P}\otimes\mathbf{I}_R)\bigg)^{-1}\;,
    \end{equation*}
    \begin{equation*}
        \boldsymbol{\mu}_B = \frac{1}{\sigma^2} \mathbf{\Sigma}_B (\mathbf{\Theta}^T \otimes \mathbf{V})^T\textnormal{Vec}(\mathbf{X})\;.
    \end{equation*}
    
    \item Given $\mathbf{V}$, $\mathbf{U}^T = \mathbf{B} \mathbf{\Theta}$, and $\sigma_B^2$, draw $\sigma^2$ via the inverse-gamma full conditional 
    \begin{equation*}
        \sigma^2\sim \textnormal{IG}\left(\tfrac{N K}{2}+a_\sigma,\;\; \tfrac12||{\mathbf{X}}-\hat{{\mathbf{X}}}||^2_{F}+b_\sigma\right)\;,
    \end{equation*}
    where $\hat{{\mathbf{X}}}=\mathbf{V}\mathbf{U}^T$ is calculated using the values of $\mathbf{V}$ and $\mathbf{U}=(\mathbf{B} \mathbf{\Theta})^T$ drawn in the previous steps, and $||\cdot||_F$ is the Frobenius norm.

    \item Given $\mathbf{B}$, draw $\sigma_B^2$ via inverse-gamma full conditional 
    \begin{equation*}
        \sigma_B^2\sim \textnormal{IG}\left(\tfrac{qR}{2}+a_B,\;\; \tfrac{1}{2}(\textnormal{Vec}(\mathbf{B}))^T(\mathbf{P}\otimes\mathbf{I}_R)\textnormal{Vec}(\mathbf{B})+b_B\right)\;.
    \end{equation*}
\end{itemize}

\begin{algorithm}
	\caption{Bayesian multiple imputation for single functional data} 
 \label{alg1}
	\begin{algorithmic}[1]
        \State \textbf{Input}: observed functional data $\{\mathbf{X}, \mathcal{O}=1\}$; prespecified rank $R$; spline basis matrix $\mathbf{\Theta}$ and the penalty matrix $\mathbf{P}$.
        \State Run the frequentist \texttt{face.sparse} algorithm to $\{\mathbf{X}, \mathcal{O}=1\}$ with $R$ principal components and get the estimated eigenfunctions $\hat{\mathbf{U}}_{\textnormal{init}}$ and scores $\hat{\mathbf{V}}_{\textnormal{init}}$. 
        \State Set the initial value of $\mathbf{U}$ as $\hat{\mathbf{U}}_{\textnormal{init}}$ and $\mathbf{V}$ as $\hat{\mathbf{V}}_{\textnormal{init}}$. Impute the missing elements of $\{\mathbf{X}, \mathcal{O}=0\}$ using $\mathbf{V}\mathbf{U}^T$. Set the initial value of $\sigma^2$ to $1$.

        \For {$s=1,...,S$-th MCMC iteration}
            \State Sample $\hat{\mathbf{V}}^s$, $\hat{\mathbf{B}}^s$, $(\hat{\sigma}^2)^{s}$, and $(\hat{\sigma}_B^2)^s$ with the conditional posterior distributions.
            \State Calculate $(\hat{\mathbf{U}}^s)^T = \hat{\mathbf{B}}^s\mathbf{\Theta}$.
            \State Calculate the underlying structure of ${\mathbf{X}}$ as $\Tilde{{\mathbf{X}}}^s = \hat{\mathbf{V}}^s(\hat{\mathbf{U}}^s)^T$.
            \State Impute the missing value $\{\hat{{\mathbf{X}}}^s, \mathcal{O}=0\}$ following the Gaussian distribution with mean $\Tilde{{\mathbf{X}}}^s$ and variance $(\hat{\sigma}^2)^{s}$.
            \State Rescale $\hat{\mathbf{U}}^s$ and $\hat{\mathbf{V}}^s$ such that each eigenfunction has unit norm.
        \EndFor		
        \State \textbf{Output}: Bayesian posterior draws of the imputed functional data $\{\hat{{\mathbf{X}}}^{1},...,\hat{{\mathbf{X}}}^S\}$.
	\end{algorithmic} 
\end{algorithm}

\vspace{-2ex}
\subsection{Bayesian Functional Imputation Algorithm}
The proposed BAMIFun algorithm is presented in Algorithm~\ref{alg1}. The algorithm requires a prespecified rank $R$.
We initialize $\mathbf{U}$ using the first $R$ eigenfunctions obtained from a frequentist FPCA, and we initialize $\mathbf{V}$ using the corresponding scores. These initial values are chosen only to facilitate faster convergence of the Bayesian algorithm; alternative initial values may also be used without affecting the validity of the method. For each MCMC iteration $s = 1, \ldots, S$, the algorithm sequentially samples $\hat{\mathbf{V}}^s$, $\hat{\mathbf{B}}^s$, $(\sigma^2)^s$, and $(\sigma_B^2)^s$ from their full conditional posterior distributions derived in Section \ref{sec:3.2}. The $\mathbf{U}$ matrix is then computed as $(\hat{\mathbf{U}}^s)^T = \hat{\mathbf{B}}^s \mathbf{\Theta}$, and the underlying smooth structure of $\mathbf{X}$ is constructed as $\tilde{\mathbf{X}}^s = \hat{\mathbf{V}}^s(\hat{\mathbf{U}}^s)^T$. The missing entries of $\mathbf{X}$ are imputed using the corresponding entries of $\hat{\mathbf{X}}^s$, where each element of $\hat{\mathbf{X}}^s$ follows a Gaussian distribution with mean $\tilde{\mathbf{X}}^s$ and variance $(\hat{\sigma}^2)^s$ at the $s$-th iteration. The imputed data are then treated as ``observed'' in the next MCMC iteration. To improve numerical stability, we further rescale the matrices $\hat{\mathbf{U}}^s$ and $\hat{\mathbf{V}}^s$ with $\hat{\mathbf{U}}^s = \hat{\mathbf{U}}^s \mathbf{D}$ and $\hat{\mathbf{V}}^s = \hat{\mathbf{V}}^s \mathbf{D}^{-1}$, where $\mathbf{D}$ is an $R\times R$ diagonal matrix such that each column of the rescaled $\hat{\mathbf{U}}^s$ has unit norm. The algorithm yields a collection of posterior draws of the imputed functional data, ${\hat{\mathbf{X}}^{1}, \ldots, \hat{\mathbf{X}}^S}$, which naturally capture the uncertainty appropriate for the multiple imputation. 
Since Algorithm \ref{alg1} requires a pre-specified rank $R$, we propose to select it using cross-validation. 
 The default parameter grid {is} set to $R\in \{R_{\textnormal{freq}}-2, R_{\textnormal{freq}}-1, R_{\textnormal{freq}}\}$ for cross-validation search; users may change the grid based on their specific data structure. The cross-validation algorithm is provided in the Supplementary Material Section 3.

\vspace{-3ex}  
\subsection{Statistical Inference for Imputation and Downstream Analysis}\label{sec:3.4}
For subject $i$ at time point $t_k$ with $\mathcal{O}_{ik}=0$, we propose to construct the $95\%$ credible interval for the functional observations as $[\hat{X}^{\textnormal{L}}_{i}(t_k), \hat{X}^{\textnormal{U}}_{i}(t_k)]$ where $\hat{X}^{\textnormal{U}}_{i}(t_k)$ is the $0.975$ quantile of the posterior samples $\{\hat{{X}}^{1}_i(t_k),\ldots,\hat{{X}}^{S}_i(t_k)\}$ and $\hat{X}^{\textnormal{L}}_{i}(t_k)$ is the $0.025$ quantile of the posterior samples. 
For downstream analysis, we take the Scalar-on-Function Regression (SoFR) as an example, but our methods are generally applicable to other statistical analyses. The point estimate $\hat{\beta}(t)$ for the functional coefficient can be obtained by pooling estimates from separate SoFR models. 
The confidence interval is then constructed following the Rubin's rule \citep{rubin1996multiple}. Specifically, for each posterior draw of the imputed dataset $\hat{{\mathbf{X}}}^{s}$, running SoFR provides the estimated functional coefficient $\hat{\beta}^s(t)$ and the corresponding variance $\hat{V}^s(t)$. We then calculate the within-imputation variance as $\hat{V}_W(t) = \frac{1}{S}\sum_{s=1}^S \hat{V}^s(t)$
and the between-imputation variance as $\hat{V}_B(t) = \frac{1}{S-1}\sum_{s=1}^S (\bar{\beta}(t)-\hat{\beta}^s(t))^2$,
where $\bar{\beta}(t) = \frac{1}{S}\sum_{s=1}^S \hat{\beta}^s(t)$ is the pooled estimate of the coefficient. The total variance at time $t$ becomes: $\hat{V}_T(t) = \hat{V}_W(t) +(1+\frac{1}{S})\hat{V}_B(t)$.
The 95\% confidence interval for $\beta(t)$ is
$\bar{\beta}(t)\pm t_{\nu,0.975}\sqrt{\hat{V}_T(t)}$, where $t_{\nu,0.975}$ is the
0.975 quantile of the $t$-distribution with degrees of freedom
$\nu = (S-1)\left(1 + \dfrac{\hat{V}_W(t)}{(1+\frac{1}{S})\hat{V}_B(t)}\right)^2$.

\vspace{-4ex}
\section{Multiway Functional Data}\label{sec:4}
The multiway functional data $\mathcal{X} = \{X_{ij}(t_k)\}_{i=1,\ldots,N;\, j=1,\ldots,J;\, k=1,\ldots,K}$ can be represented as a tensor of dimension $N\times J\times K$. We adopt the FTSVD model proposed by \cite{han2024guaranteed}, which approximates the multiway functional data through the decomposition
\begin{equation}\label{eq:multilevel_model}
    \mathcal{X} = \sum_{r=1}^R \lambda_r \mathbf{v}_r \circ \mathbf{w}_r \circ \mathbf{u}_r + \mathcal{E}\;.
\end{equation}
For each $r=1,\ldots,R$, $\mathbf{v}_r$ and $\mathbf{w}_r$ are unit-norm vectors of length $N$ and $J$ corresponding to the subject and feature dimensions, respectively. $\mathbf{u}_r = u_r(t)$, $t=t_1,\ldots,t_K$, represents the functional mode on which we impose a smoothness constraint. Model~\eqref{eq:multilevel_model} is similar to the CP decomposition of tensor data proposed in \cite{jiang2025bamita}, except that a smoothness constraint is imposed on $\mathbf{u}_r$.
Unlike \citet{han2024guaranteed}, which develops a frequentist estimation procedure for the FTSVD model, we propose a Bayesian sampling algorithm that facilitates simultaneous parameter estimation and multiple imputation of sparsely observed functional trajectories.
We emphasize that the FTSVD model for the multiway functional data is different from the multilevel FPCA (MFPCA) model in \citet{di2009multilevel}; see detailed discussion in the Supplementary Material Section 5.1.


Define the matrix $\mathbf{V}=[\mathbf{v}_1\;\;\mathbf{v}_2\;\cdots\;\mathbf{v}_R]$, $\mathbf{W}=[\mathbf{w}_1\;\;\mathbf{w}_2\;\cdots\;\mathbf{w}_R]$, $\mathbf{U}=[\mathbf{u}_1\;\;\mathbf{u}_2\;\cdots\;\mathbf{u}_R]$, and the notation $\llbracket \mathbf{V}, \mathbf{W}, \mathbf{U}\rrbracket \coloneqq \sum_{r=1}^R \mathbf{v}_r\circ \mathbf{w}_r\circ \mathbf{u}_r$. We can re-express model \eqref{eq:multilevel_model} in matrix form as $\mathcal{X}_{(3)} = \mathbf{U}\mathbf{\Lambda} (\mathbf{W}\odot\mathbf{V})^T + \mathcal{E}_{(3)}$ \citep{jiang2025bamita},
where $\mathcal{X}_{(3)}$ is the mode-3 matricization of tensor $\mathcal{X}$ with dimension $K\times NJ$, $\mathbf{\Lambda}$ is a diagonal matrix of $(\lambda_1,...,\lambda_R)$, 
and each element of $\mathcal{E}_{(3)}$ follows an i.i.d. Gaussian distribution with mean $0$ and variance $\sigma^2$. 
In the MCMC sampling, we enforce $\mathbf{v}_r$, $\mathbf{w}_r$, and $\mathbf{u}_r$ to have the same norm of $\lambda_r^{1/3}$, which makes $\mathbf{\Lambda}=\mathbf{I}_R$ where $\mathbf{I}_R$ is the identity matrix with dimension $R$. Thus, the tensor matricization $\mathcal{X}_{(3)}$ becomes $\mathcal{X}_{(3)} = \mathbf{U} (\mathbf{W}\odot\mathbf{V})^T + \mathcal{E}_{(3)}$.
Similarly, matricizing the tensor $\mathcal{X}$ along other dimensions yields $\mathcal{X}_{(1)} = \mathbf{V} (\mathbf{U}\odot\mathbf{W})^T+ \mathcal{E}_{(1)}$ and $\mathcal{X}_{(2)} = \mathbf{W} (\mathbf{U}\odot\mathbf{V})^T+ \mathcal{E}_{(2)}$.
Similar to the single-level functional data, we use basis expansions to model the matrix $\mathbf{U}^T = \mathbf{B\Theta}$ and impose a shrinkage prior on the parameter matrix $\mathbf{B}$ as in \eqref{eq:priorforB}. We place proper priors on the remaining parameters, $\mathbf{V}_{ir}\stackrel{\textnormal{iid}}{\sim}N(0,\sigma_V^2)$, $\mathbf{W}_{jr}\stackrel{\textnormal{iid}}{\sim}N(0,\sigma_W^2)$, $\sigma^2\sim\textnormal{IG}(a_\sigma,b_\sigma)$, and $\sigma_B^2\sim\textnormal{IG}(a_B,b_B)$, mirroring the single-level specification.

The final Bayesian model for multiway functional data imputation becomes
\begin{equation}\label{eq:multiway.model}
    \begin{cases}
      \qquad \textnormal{Vec}(\mathcal{X}_{(3)}^T) = (\mathbf{\Theta}^T \otimes (\mathbf{W}\odot\mathbf{V}))\textnormal{Vec}(\mathbf{B})+  \textnormal{Vec}(\mathcal{E}_{(3)}^T)\;;\\
      \qquad \textnormal{Vec}(\mathcal{E}_{(3)}^T)\sim N(0, \sigma^2\mathbf{I}_{NJK})\;; \\
      \qquad \textnormal{Vec}(\mathbf{B}) \sim p(\textnormal{Vec}(\mathbf{B})),\; \;\; \mathbf{W} \sim p(\mathbf{W}),\; \;\; \mathbf{V} \sim p(\mathbf{V}),\;\;\; \sigma^2\sim p(\sigma^2),\;\textnormal{and} \;\;\; \sigma_B^2\sim p(\sigma_B^2).
    \end{cases}       
\end{equation}
We matricize the tensor along the third dimension because it corresponds to the functional observations. The Gibbs sampler for multiway BAMIFun under \eqref{eq:multiway.model} is presented in the Supplementary Material Section~5.2.

\vspace{-4ex}
\section{Simulation Experiments}\label{sec:5}

Across all three simulation experiments in this section, each scenario is evaluated over $500$ replications. We run two separate MCMC chains, each with $1000$ iterations for BAMIFun and $3000$ iterations for BAMITA (to improve convergence of BAMITA), with the first $500$ iterations discarded as burn-in. To assess convergence, we compute the potential scale reduction factor (PSRF) \citep{gelman1992inference} on the element-wise low-rank reconstruction; PSRF values close to $1$ indicate successful convergence. 
The complete PSRF summaries and per-replication runtimes are reported in the Supplementary Material Section 6. 

\subsection{Imputation Performance for Single-Level Functional Data}
\subsubsection{Simulation settings}
In the first simulation experiment, we evaluate the imputation performance for the single-level functional data. We generate functional data $\{X_i(t),t\in \{ \frac{1}{K},...,\frac{K}{K}\}\}$ with $K=100$ from $ X_i(t)  = \sum_{h=1}^{12}\xi_{ih}u_h(t) +\epsilon_i(t)$,
where $u_h$ are the 12 eigenfunctions obtained from our NHANES data application, and $\xi_{ih}\sim N(0, \lambda_h)$ are the scores with eigenvalues $\{\lambda_h\}_{h=1}^{12} = [2,2,2,1,1,1,0.5,0.5,0.5,$ $0.1,0.1,0.1]$. We draw each $\epsilon_i(t)$ independently from a standard Gaussian distribution. For each subject $i=1,...,N$ where $N\in\{100, 300, 500, 1000\}$, we randomly set a proportion $\rho \in\{ 0.8, 0.9, 0.95\}$ of their functional observations as missing ({corresponding to $20$, $10$, and $5$ observations per subject respectively}). 

For each simulation scenario, we compare our BAMIFun algorithm with (1) the PACE algorithm implemented using the state-of-the-art fast sparse FPCA \citep{xiao2018fast}, and (2) the Bayesian multiple imputation algorithm for tensor array (BAMITA) \citep{jiang2025bamita}. In all experiments, we set the proportion of variance explained (pve) as 0.99 for PACE. For BAMIFun, we use the B-spline basis of dimension $L=15$ to model the functional space.  As is standard in penalized splines, $L$ only needs to be sufficiently large to provide a flexible representation of the underlying functions, while the roughness penalty controls the effective smoothness of the fitted curves. Once $L$ is sufficiently large, further increases in $L$ have little impact on the fitted functions \citep{ruppert2002selecting}. 
The tuning parameter $R$ {is} selected through CV over the default parameter grid, { where} $R_{\textnormal{freq}}$ is the number of principal components estimated through FPCA. 
We evaluate the imputation performance for the single-level functional data in terms of: (1) the relative mean squared error (RMSE) for the imputation $\frac{\sum_{\mathcal{O}=0}((\textbf{X}-\hat{\textbf{X}})^2)}{\sum_{\mathcal{O}=0}(\textbf{X}^2) }$,
where $\hat{\textbf{X}}$ is the imputed functional data, and $\mathcal{O}=0$ indicates unobserved points; and (2) the coverage rate for the estimated 95\% confidence/credible intervals for the imputed functional data. The confidence interval for PACE is calculated using the estimated standard error from \texttt{face::face.sparse()}.

\vspace{-3ex} 
\subsubsection{Results}
Figure \ref{fig:impu} reports the imputation RMSE and coverage rates across methods. For RMSE, BAMIFun performs comparably to PACE when there are 20 observations per subject (80\% missingness). At 10 observations per subject (90\% missingness), it shows slightly larger {RMSE} for small sample sizes ($N=100$), but the difference diminishes as $N$ increases. Under extremely high missingness (5 observations per subject), BAMIFun initially has a higher {RMSE} than PACE; however, the gap narrows with larger sample sizes. Across all settings, incorporating the smoothness constraint allows BAMIFun to consistently outperform BAMITA  in terms of imputation accuracy. 
For 95\% confidence/credible intervals, BAMIFun maintains a coverage rate close to nominal, which is substantially higher than PACE. BAMITA exhibits under-coverage when the sample size is small, likely due to its larger RMSE and poor convergence (see the PSRF figure in the Supplementary Material Section 6.1). 
In the Supplementary Material Section 6.2, we present results using a more informative $\textnormal{IG}(2,1)$ prior for the smoothing parameter $\sigma_B^2$, which are similar to those under the default $\textnormal{IG}(0.01,0.01)$ prior.

In the Supplementary Material Section 6.3, we present an additional experiment that fixes the number of observations per subject at $10$ while refining the grid from {$K=50$ to $K=400$}. As {$K$} grows, BAMITA's imputation RMSE and coverage rate deteriorate sharply, since it treats time as an unordered mode and cannot borrow information from neighboring points. In contrast, the RMSE and coverage of BAMIFun and PACE remain largely unchanged.

\begin{figure}[ht] 
\centering
\includegraphics[width=\textwidth]{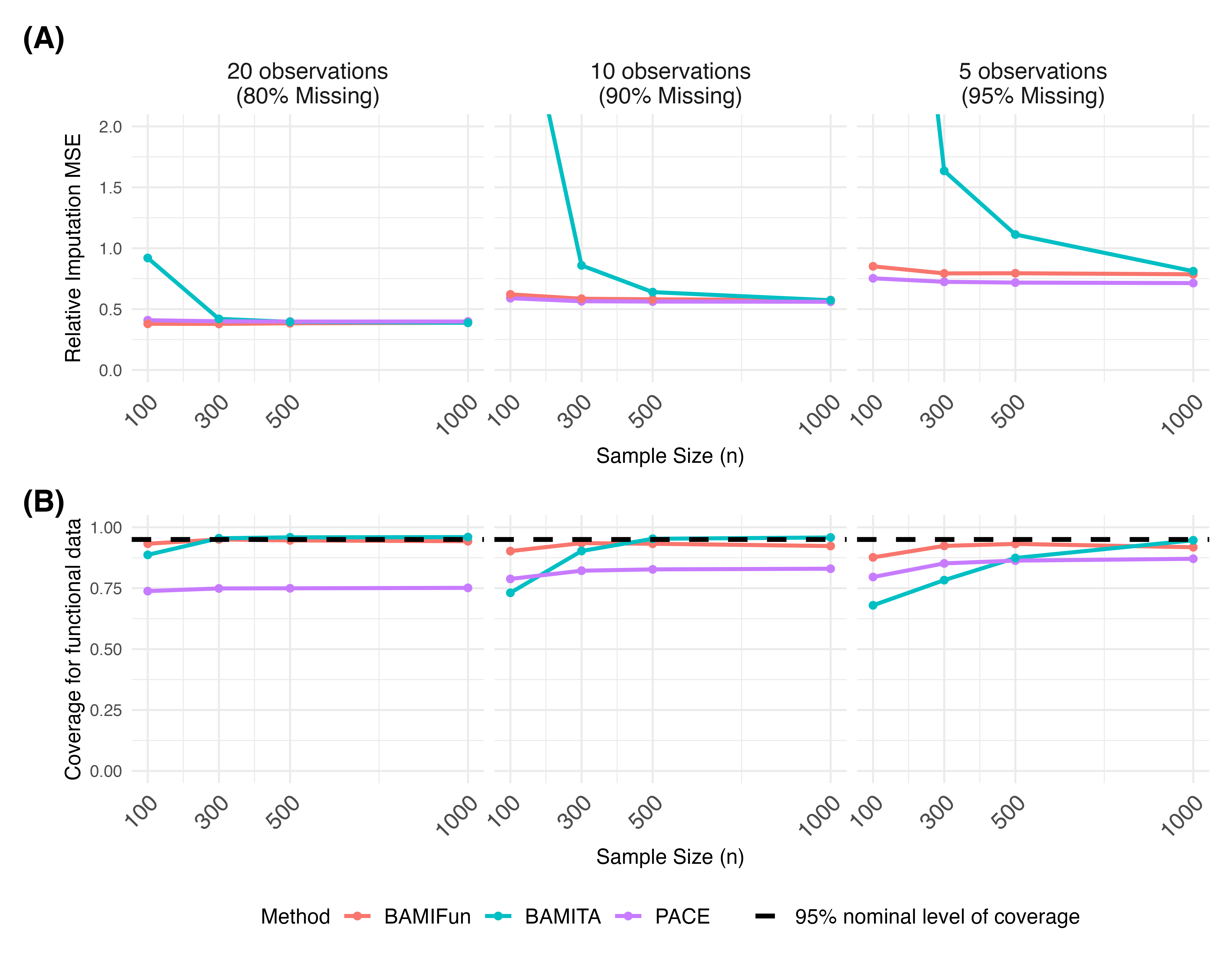}
\caption{Simulation results for the imputation performance of single-level functional data. Panel (A) displays the relative imputation {RMSE}, and Panel (B) displays the coverage rate. We vary the missing proportion and sample size for the data-generating mechanism.}
\label{fig:impu}
\end{figure}

\subsection{Downstream Analysis with Imputed Functional Data}

\subsubsection{Simulation settings}
In the second simulation experiment, we compare imputation algorithms in terms of the inferential performance in downstream analysis using the imputed data. We generate functional data following the same mechanism as in the first experiment. 
An outcome variable is then generated from a Gaussian distribution with variance 1 and mean $\mathbb{E}[Y_i] = \int_0^1 \beta(t)X_i(t)dt$,
where $\beta(t) = -5t^2+5t+0.17$, $t\in[0, 1]$ and $X_i(t)$ is the complete simulated functional data. 

For downstream analysis, we first impute the functional data using each of the three methods (PACE, BAMITA, BAMIFun).
A scalar-on-function regression (SoFR) model is then fit on the imputed functional data using the R package \texttt{mgcv}. For each SoFR model, pointwise 95\% confidence intervals for $\hat{\beta}(t)$ are constructed based on a nonparametric bootstrap ($B=100$).
For PACE, we report coverage of the pointwise $2.5\%$ and $97.5\%$ bootstrap quantiles. For BAMIFun and BAMITA, the $S$ sets of estimates are pooled using Rubin's rules described in Section~\ref{sec:3.4} to construct 95\% pointwise confidence intervals.
We calculate the relative integrated squared error (RISE) as $\frac{\int_0^1 (\hat{\beta}(t)-\beta(t))^2dt}{\int_0^1 \beta(t)^2dt}$. We also calculate the empirical coverage rate of the estimated 95\% confidence interval for the functional coefficient. 


\vspace{-3ex}
\subsubsection{Simulation results}
Results are presented in Figure~\ref{fig:sofr}. In terms of the RISE, BAMIFun demonstrates good performance across simulation scenarios (except for one fluctuation under $95\%$ missingness). 
With respect to coverage, BAMIFun consistently outperforms PACE, which exhibits systematic undercoverage by ignoring the imputation uncertainty. { Although BAMITA attains comparable coverage, its downstream point estimation is far less stable than BAMIFun: the RISE of the functional coefficient is more variable across replications, with occasional severe failures when data are scarce.}
These findings together underscore the advantage of BAMIFun, which appropriately accounts for imputation uncertainty and leads to more reliable interval estimation in downstream analyses.

\begin{figure}[ht] 
\centering
\includegraphics[width=\textwidth]{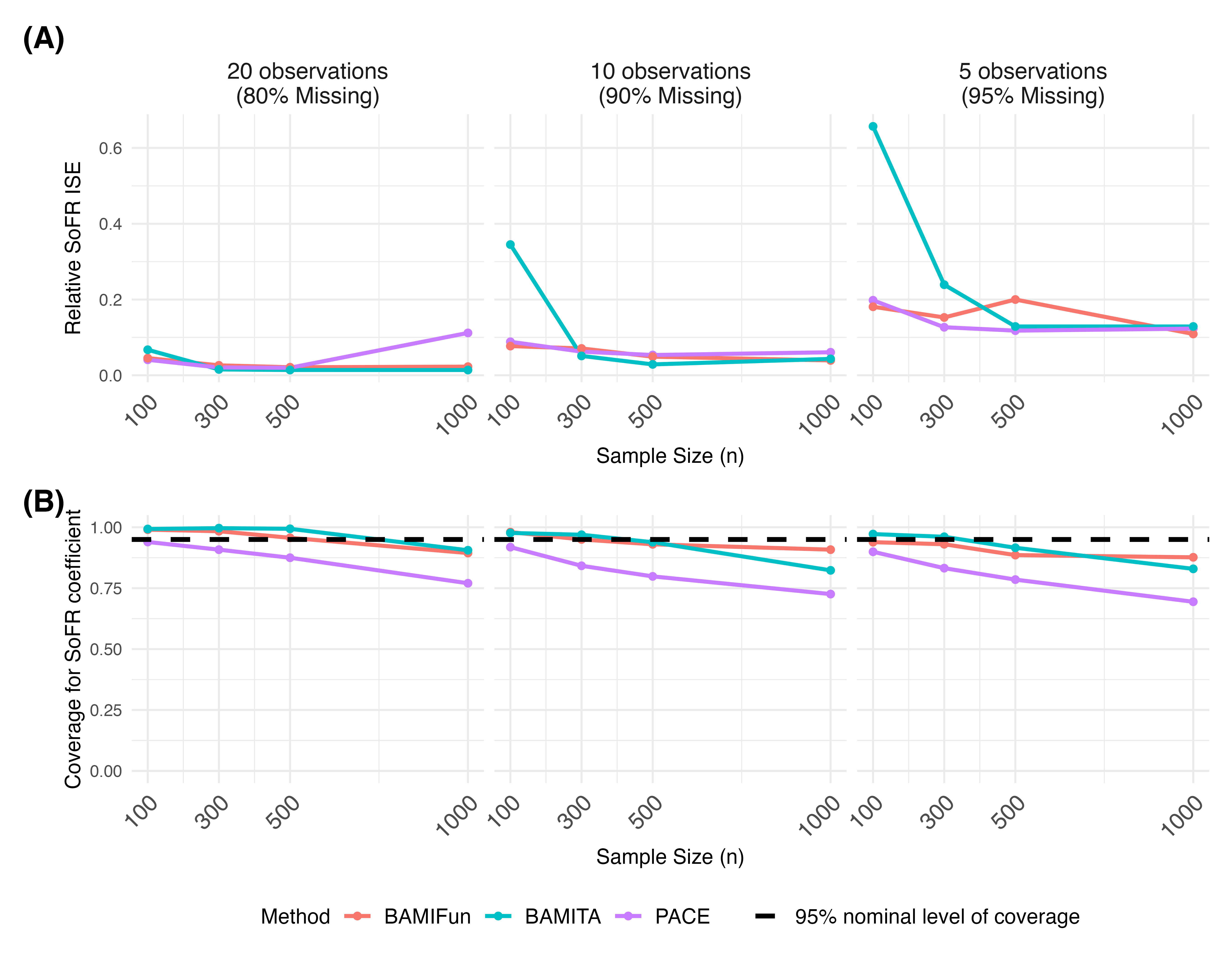}
\caption{Simulation results for the performance of downstream Scalar-on-Function Regression (SoFR) using the imputed single-level functional data. Panel (A) displays the relative integrated squared error ({RISE}) of the estimated SoFR coefficient, and Panel (B) displays the coverage rate. We vary the missing proportion and sample size for the data-generating mechanism.}
\label{fig:sofr}
\end{figure}

\subsection{Imputation Performance for Multiway Functional Data}
\subsubsection{Simulation settings}
In the third simulation experiment, we evaluate BAMIFun for the multiway functional data $\mathcal{X}=\{X_{ij}(t),t\in \{ \frac{1}{K},...,\frac{K}{K}\}, i=1,...,N\}$, with $K=100$ and the number of visits $J=4$. We generate multiway functional data based on whether a low-rank structure exists. For the scenario with a low rank structure, $\mathcal{X}=\sum_{r=1}^4 a_r\circ b_r\circ u_r + \mathcal{E}$,
where $a_r$ and $b_r$ are vectors with length $N$ and $J$ whose elements follow an i.i.d. Gaussian distribution with mean $0$ and variance $\{\sigma_r^2\}_{r=1}^4= \{2, 1, 0.5, 0.1\}$, $u_r$ are the same set of eigenfunctions as in the first simulation study, and the elements in $\mathcal{E}$ follow an i.i.d. Gaussian distribution with mean $0$ and variance $1$. All other settings are the same as in the previous simulation. { In the Supplementary Material Section 6.5, we present results for the scenario without a low rank structure.} 
{
We compare BAMIFun with PACE based on the frequentist MFPCA \citep{di2009multilevel} and with BAMITA \citep{jiang2025bamita}. 
For Bayesian approaches, we compute both the RMSE and the empirical coverage of imputed elements. For PACE, we report only the RMSE since existing R implementations do not provide confidence intervals}.

\subsubsection{Results}

The simulation results are summarized in Table~\ref{tab:c7v2}. Compared with PACE (MFPCA), BAMIFun achieves lower RMSE across all settings. Compared with BAMITA, BAMIFun also attains lower RMSE in most settings, with the largest gains at small samples and under severe missingness. This reflects the stabilizing effect of smoothness constraints in situations where data are scarce. For larger samples, BAMITA and BAMIFun yields similar RMSEs. In terms of coverage rate, BAMIFun maintains coverage close to the nominal level across all settings, whereas BAMITA exhibits under coverage when the sample size is small.

\begin{table}[ht]
\centering
\setlength{\tabcolsep}{6pt}
\caption{Imputation performance on low-rank functional data: relative imputation {RMSE} with the mean 95\% element-wise coverage (in parentheses, \%) for PACE, BAMIFun, and BAMITA, across the number of observed timepoints per subject and the sample size. Coverage is not reported for PACE.}
\label{tab:c7v2}
\begin{tabular}{cc ccc}
\toprule
Sample & Observations & \multicolumn{3}{c}{Method} \\
\cmidrule(lr){3-5}
Size & (Missing) & PACE & BAMIFun & BAMITA \\
\midrule
100 & 20 (80\%) & 0.237($-$) & 0.194(94.6) & 0.200(94.4) \\
 & 10 (90\%) & 0.430($-$) & 0.211(94.5) & 0.253(93.8) \\
 & 5 (95\%) & 0.614($-$) & 0.261(93.9) & 0.782(90.1) \\
\addlinespace
300 & 20 (80\%) & 0.227($-$) & 0.191(94.6) & 0.188(94.6) \\
 & 10 (90\%) & 0.385($-$) & 0.220(94.5) & 0.217(94.4) \\
 & 5 (95\%) & 0.582($-$) & 0.250(94.2) & 0.277(93.7) \\
\addlinespace
500 & 20 (80\%) & 0.233($-$) & 0.200(94.6) & 0.195(94.6) \\
 & 10 (90\%) & 0.382($-$) & 0.219(94.5) & 0.222(94.4) \\
 & 5 (95\%) & 0.581($-$) & 0.255(94.2) & 0.259(93.9) \\
\addlinespace
1000 & 20 (80\%) & 0.218($-$) & 0.196(94.6) & 0.195(94.6) \\
 & 10 (90\%) & 0.373($-$) & 0.217(94.5) & 0.220(94.5) \\
 & 5 (95\%) & 0.565($-$) & 0.241(94.2) & 0.253(94.1) \\
\bottomrule
\end{tabular}
\end{table}


\vspace{-8ex}
\section{Case Studies}\label{sec:6}
\subsection{Single-level Functional Data}
For the single-level case study, we apply BAMIFun to the physical activity data in the National Health and Nutrition Examination Survey (NHANES) and compare the imputation performance against PACE and BAMITA. NHANES is a nationwide program conducted by the U.S. Centers for Disease Control and Prevention to monitor the health and nutritional status of adults and children in the United States \citep{cui2021additive, cui2022semiparametric}. For this analysis, we use the accelerometry data collected from the 2011–2012 and 2013–2014 cycles of NHANES and organized by \citet{crainiceanu2024functional}. 

We include participants who were older than 50 years at the time of enrollment with a total of $N=3880$. The functional measurements are Monitor-Independent Movement Summary (MIMS) values, a device-independent measure of activity intensity \citep{john2019open}; the released MIMS data originally had $K=1440$ time points.  For each subject, the measurements are averaged across all available days and then uniformly subsampled to $720$ time points to reduce computational cost. Missingness is introduced by randomly masking $97.5\%, 95\%$, or $90\%$ of the observed time points. We then apply the imputation methods to the resulting incomplete curves and compute the imputation RMSE and coverage for the removed entries for both BAMIFun and PACE. Similar to simulation experiments, for PACE we select the number of principal components by setting \texttt{pve = 0.95}. For BAMIFun, we use the same number of principal components as in PACE.

{We repeat each scenario 100 times and summarize the RMSE and coverage in Table~\ref{tab:casestudy1}. BAMIFun yields a slightly higher RMSE under extreme missingness, but the gap diminishes when the missing proportion decreases to approximately 95\%. However, the inferential performance is notably different. Across all levels of missingness, BAMIFun provides empirical coverage close to the nominal level for the estimated 95\% credible intervals. In contrast, PACE exhibits substantially lower coverage across all settings. A downstream scalar-on-function analysis on this dataset is reported in the Supplementary Material Section~7.1.}

\begin{table}[ht]
\centering
\caption{\textnormal{Imputation performance with the NHANES dataset for the PACE, BAMIFun, and BAMITA algorithms. For
each method, we report the {RMSE} and the coverage rate (in parentheses)}}
\renewcommand{\arraystretch}{1.1} 
\begin{tabular}{cccc}
\toprule
\textbf{Missing} &
\textbf{PACE} &
\textbf{BAMIFun} &
\textbf{BAMITA}\\
\textbf{Proportion} &
\textbf{RMSE (Coverage)} &
\textbf{RMSE (Coverage)} & 
\textbf{RMSE (Coverage)}\\
\midrule
 97.5\% & 0.140(71.4) & 0.148(92.3) & 0.173(90.1)\\
 95\% & 0.124(61.9) & 0.127(93.4) & 0.137(92.7)\\
 90\% & 0.115(53.4) & 0.117(93.7) & 0.121(93.5)\\
\bottomrule
\end{tabular}
\label{tab:casestudy1}
\end{table}

\vspace{-2ex}
\subsection{Multiway Functional Data}
The human gut harbors a complex and dynamic microbial ecosystem that evolves rapidly during early life and plays a critical role in immune development and overall health. We analyze data from a longitudinal study of 52 preterm infants admitted to the neonatal intensive care unit (NICU), in which stool samples were repeatedly collected over the first three months of life \citep{cong2017influence}. Sampling times varied across infants, yielding sparsely observed longitudinal profiles with $K=118$.
Microbial composition was quantified using 16S rRNA sequencing, and relative abundances were summarized at the genus level, resulting in 152 unique genera across all samples. The data form a three-way array indexed by subject, microbial genus, and time, motivating a multiway functional representation in which each infant is associated with a collection of genus-specific abundance trajectories over time. The resulting multiway functional dataset has dimensions $52\times 152\times 118$, of which $91.1\%$ of the entries are unobserved due to the study design, posing substantial challenges for imputation. 

We compare the imputation performance between BAMIFun and BAMITA.
Frequentist MFPCA methods are not considered, as existing software implementations cannot accommodate datasets with such an extreme level of missingness.
Because cross-validation (CV) has already been assessed in previous simulations and applications, and to reduce computational cost, we do not use CV to select the rank $R$. Instead, we vary $R$ from $10$ to $38$ and present the corresponding results. 
In each repetition, we randomly sample $30\%$ of the observed entries as a test set and apply the imputation methods to the rest of the observed entries. For each $R$, we repeat the experiment $100$ times and report the mean imputation RMSEs over the test set in Figure \ref{fig:app2}. BAMIFun consistently outperforms BAMITA. In addition, BAMIFun achieves coverage rates closer to the nominal level than BAMITA across the values of $R$.

\begin{figure}[ht] 
\centering
\includegraphics[width=\textwidth]{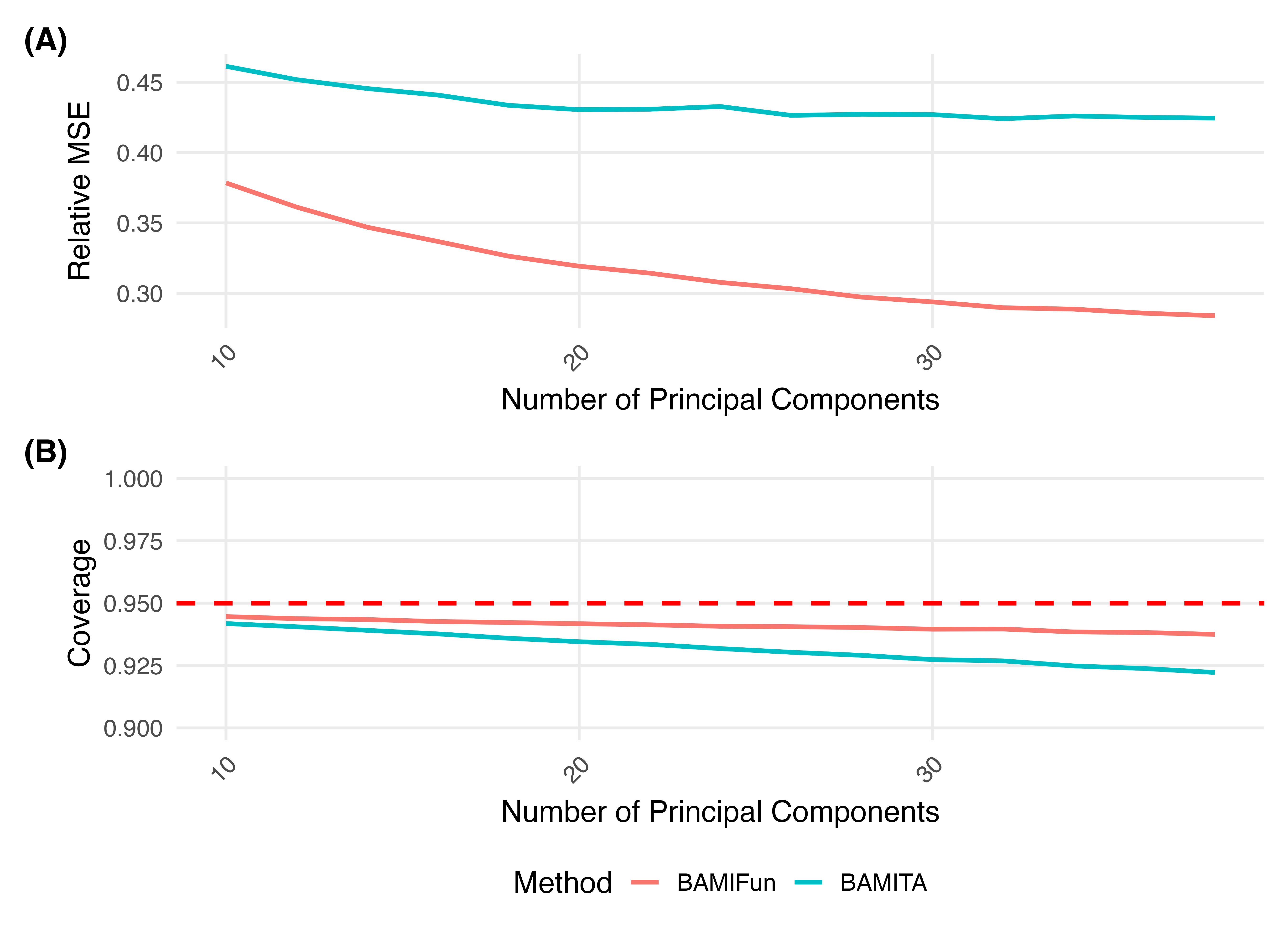}
\caption{Multiple imputation results using the infant gut microbiome dataset. Panel (A) displays the {relative imputation RMSE} for the imputed elements. In Panel (B), we display the coverage rate for the imputed elements.}
\label{fig:app2}
\end{figure}

\vspace{-4ex}
\section{Discussion}\label{sec:7}
We developed BAMIFun, a Bayesian multiple imputation framework for sparse functional data, including both single-level and multiway functional settings. BAMIFun provides a principled approach to reconstructing incomplete trajectories while explicitly quantifying the uncertainty arising from the imputation. 
Our method addresses a key limitation of the widely used PACE approaches, that is, their tendency to treat imputed curves as known, thereby underestimating uncertainty in downstream analyses. 
For single-level functional data, BAMIFun achieves similar imputation accuracy compared with PACE. However, by effectively accounting for the imputation uncertainty, BAMIFun achieves nominal coverage and provides much better inferential performance for downstream analysis. 
For multiway functional data, we adopted the FTSVD model \citep{han2024guaranteed} and proposed a novel Bayesian imputation algorithm. Since this model differs from the MFPCA model, the imputation accuracy for the two methods depends on the specific data structure, as demonstrated in simulations.

For single-level functional data, BAMIFun and PACE both target the rank-$R$ smooth reconstruction $\mathbf{V}\mathbf{U}^{T}$, with subtle differences: BAMIFun fits an unconstrained low-rank factorization with a smoothness penalty and does not require the columns of $\mathbf{U}$ to be orthogonal. Since the imputation depends only on $\mathbf{V}\mathbf{U}^{T}$, we do not expect BAMIFun to improve imputation accuracy compared to PACE.
What BAMIFun gains is the substantially improved coverage and valid downstream inference, as confirmed in our simulations and applications. For multiway functional data, the performance instead depends on the structure, as BAMIFun is more accurate when a low-rank structure is present.
Furthermore, BAMIFun attains close-to-nominal coverage and remains applicable under extreme missingness, whereas existing MFPCA software cannot run.

This paper assumes that the observation pattern may depend on the observed functional values but not on the unobserved ones, which covers MCAR as in our simulations. When missingness instead depends on other auxiliary information, such as scalar outcomes that are informative about the unobserved curve, methods that explicitly model that information \citep{petrovich2022highly, jang2021bayesian} are more appropriate. Extensions of BAMIFun to these more complex dependency settings will be explored in the future.

\section*{Supplementary Materials}
Web Appendices, Tables, and Figures referenced in Sections~3, 4, 5, and~6, and the R code implementing the proposed method and reproducing the simulation studies and data analyses, are available with this paper at the Biometrics website on Oxford Academic. The R code is also available on GitHub at \url{https://github.com/ZirenJiang/BAMIFun}.

\section*{Data Availability}
The physical activity data that support the findings in this paper were obtained from the National Health and Nutrition Examination Survey 2011--2014 and are publicly available from the National Center for Health Statistics, Centers for Disease Control and Prevention. For more detail, please visit \url{https://www.cdc.gov/nchs/nhanes}.
The infant gut microbiome data that support the findings in this paper are not publicly available due to privacy or ethical restrictions.

\section*{Funding}
This work was supported by the National Institutes of Health (R01HL187193, R21LM014912, and R35GM161467). The content is solely the responsibility of the authors and does not necessarily represent the views of the National Institutes of Health.

\phantomsection
\addcontentsline{toc}{section}{References}
\begingroup
\spacingset{0.5}

\par
\endgroup

\clearpage
\StartSupplement
\pdfbookmark[0]{Supplementary Materials}{supplement}
\begin{center}
{\LARGE\setstretch{1.1} Supplementary Materials for ``Bayesian Multiple Imputation for Functional Data'' by Ziren Jiang, Lei Xuan, Eric F. Lock, and Erjia Cui\par}
\end{center}

\renewcommand{\thealgorithm}{S\arabic{algorithm}}
\section{Bayesian Multiple Imputation Algorithm for Functional Data}
\label{sec:supp-model}

\subsection{The Bayesian Imputation Model}
\label{sec:supp-model-spec}
Recall the low-rank model with the penalized-spline representation of the functional mode,
\begin{equation}\label{eq:supp-model}
    \mathbf{X} = \mathbf{V}\mathbf{B}\mathbf{\Theta} + \mathbf{E}\;,
\end{equation}
where $\mathbf{X}\in\mathbb{R}^{N\times K}$, $\mathbf{V}\in\mathbb{R}^{N\times R}$, $\mathbf{B}\in\mathbb{R}^{R\times L}$, $\mathbf{\Theta}\in\mathbb{R}^{L\times K}$, and $\mathbf{U}^T=\mathbf{B}\mathbf{\Theta}$ so that $\mathbf{U}\in\mathbb{R}^{K\times R}$. The entries of $\mathbf{E}$ are i.i.d.\ $N(0,\sigma^2)$. Throughout, $\textnormal{Vec}(\cdot)$ denotes the column-wise vectorization, and we repeatedly use the identity $\textnormal{Vec}(\mathbf{A}\mathbf{B}\mathbf{C})=(\mathbf{C}^T\otimes\mathbf{A})\textnormal{Vec}(\mathbf{B})$. Vectorizing both sides of \eqref{eq:supp-model} with $\mathbf{A}=\mathbf{V}$, $\mathbf{B}=\mathbf{B}$, and $\mathbf{C}=\mathbf{\Theta}$ yields the equivalent Bayesian linear model
\begin{equation}\label{eq:supp-vec}
    \textnormal{Vec}(\mathbf{X}) = (\mathbf{\Theta}^T\otimes\mathbf{V})\textnormal{Vec}(\mathbf{B}) + \textnormal{Vec}(\mathbf{E})\;,\qquad \textnormal{Vec}(\mathbf{E})\sim N(\mathbf{0},\sigma^2\mathbf{I}_{NK})\;,
\end{equation}
so that the sampling distribution of the data is
\begin{equation}\label{eq:supp-lik}
    p(\mathbf{X}\mid\mathbf{V},\mathbf{B},\sigma^2)\propto (\sigma^2)^{-NK/2}\exp\!\left(-\tfrac{1}{2\sigma^2}\big\|\textnormal{Vec}(\mathbf{X})-(\mathbf{\Theta}^T\otimes\mathbf{V})\textnormal{Vec}(\mathbf{B})\big\|^2\right)\;.
\end{equation}

\medskip
\noindent\textbf{Smoothness prior on $\mathbf{B}$.}\quad
As derived in Section~3.1 of the main manuscript, the roughness penalty on the $r$-th functional factor is $\mathbf{B}_{r\cdot}^T\mathbf{P}\mathbf{B}_{r\cdot}$, where $\mathbf{B}_{r\cdot}\in\mathbb{R}^{L}$ is the $r$-th row of $\mathbf{B}$ and $\mathbf{P}$ is the $L\times L$ penalty matrix. Summing over the $R$ factors and recalling that $\textnormal{Vec}(\mathbf{B})$ stacks the columns of $\mathbf{B}$, so the entry $\mathbf{B}_{rl}$ occupies position $(l-1)R+r$, gives
\begin{equation}\label{eq:supp-penalty}
    \sum_{r=1}^R \mathbf{B}_{r\cdot}^T\mathbf{P}\mathbf{B}_{r\cdot}
    = \sum_{r=1}^R\sum_{l_1=1}^L\sum_{l_2=1}^L \mathbf{B}_{r l_1}\,\mathbf{P}_{l_1 l_2}\,\mathbf{B}_{r l_2}
    = \textnormal{Vec}(\mathbf{B})^T(\mathbf{P}\otimes\mathbf{I}_R)\textnormal{Vec}(\mathbf{B})\;,
\end{equation}
where the last equality holds because $(\mathbf{P}\otimes\mathbf{I}_R)_{(l_1-1)R+r_1,\,(l_2-1)R+r_2}=\mathbf{P}_{l_1 l_2}\,\delta_{r_1 r_2}$: the penalty couples the basis coefficients (indexed by $l$) through $\mathbf{P}$ while acting as the identity across the factors (indexed by $r$).

The penalty matrix $\mathbf{P}$ is rank deficient: for the second-derivative penalty, its null space consists of the basis coefficients of affine functions, for which the penalty vanishes, so that $q:=\textnormal{rank}(\mathbf{P})=L-2$ and $m:=\dim\textnormal{null}(\mathbf{P})=2$. Adding the penalty \eqref{eq:supp-penalty} to the log-likelihood is equivalent to placing on $\textnormal{Vec}(\mathbf{B})$ the \emph{partially informative} Gaussian prior \citep{speckman2003fully,jiang2025tutorial}
\begin{equation}\label{eq:supp-priorB}
    p\bigl(\textnormal{Vec}(\mathbf{B})\mid\sigma_B^2\bigr)
    = \bigl|\mathbf{P}\otimes\mathbf{I}_R\bigr|_{+}^{1/2}\,
    \bigl(2\pi\sigma_B^2\bigr)^{-qR/2}
    \exp\!\left(-\tfrac{1}{2\sigma_B^2}\textnormal{Vec}(\mathbf{B})^T(\mathbf{P}\otimes\mathbf{I}_R)\textnormal{Vec}(\mathbf{B})\right),
\end{equation}
where $\sigma_B^2$ is the smoothing parameter and $|\cdot|_+$ denotes the product of the non-zero eigenvalues. Since $\mathbf{P}$ is rank deficient, \eqref{eq:supp-priorB} is not a probability density on $\mathbb{R}^{RL}$. Supplementary Material Section~\ref{sec:supp-proper} shows that \eqref{eq:supp-priorB} is exactly the product of a flat prior on the unpenalized coordinates and a proper Gaussian prior on the penalized coordinates. The power of $\sigma_B^2$ in \eqref{eq:supp-priorB} is the \emph{rank} $qR$ of the precision matrix rather than its dimension $RL$, following \citet[eq.~(9)]{speckman2003fully}; the two conventions differ only by a factor $(\sigma_B^2)^{-mR/2}$, so combining the dimension-based kernel with an $\textnormal{IG}(a_B,b_B)$ hyperprior is equivalent to combining \eqref{eq:supp-priorB} with an $\textnormal{IG}(a_B+mR/2,\,b_B)$ hyperprior, and both specifications are proper in $\sigma_B^2$ whenever $a_B,b_B>0$.

\medskip
\noindent\textbf{The remaining proper parameters.}\quad
We adopt
\begin{align}
    \mathbf{V}_{ir} &\stackrel{\textnormal{iid}}{\sim} N(0,\sigma_V^2)
    \quad\Longleftrightarrow\quad
    \boldsymbol{v}_{i\cdot}^T \stackrel{\textnormal{iid}}{\sim} N(\mathbf{0},\,\sigma_V^2\mathbf{I}_R),
    \qquad i=1,\dots,N,\ r=1,\dots,R, \label{eq:supp-pp-priorV}\\
    \sigma^2 &\sim \textnormal{IG}(a_\sigma,\,b_\sigma), \qquad a_\sigma>0,\ b_\sigma>0, \label{eq:supp-prop-priorsig}\\
    \sigma_B^2 &\sim \textnormal{IG}(a_B,\,b_B), \qquad a_B>0,\ b_B>0, \label{eq:supp-pp-priorsigB}
\end{align}
where $\sigma_V^2>0$ is a fixed constant (we take $\sigma_V^2=10^6$ by default) and $a_\sigma=b_\sigma=a_B=b_B=0.01$ by default. We use the shape--scale convention in which $\textnormal{IG}(a,b)$ has density $f_{\textnormal{IG}}(s;a,b)=\{b^{a}/\Gamma(a)\}\,s^{-a-1}\exp(-b/s)$ for $s>0$. The Gaussian prior on $\mathbf{V}$ mirrors the standard treatment of principal component scores as Gaussian random effects in joint functional models \citep[Section~4]{jiang2025tutorial}, and the inverse-gamma hyperpriors on the variance components follow common practice in Bayesian penalized splines \citep{lang2004bayesian,jiang2025tutorial}. In particular, a proper prior on $\sigma_B^2$ ($b_B>0$) is essential rather than merely convenient, since a diffuse prior on the variance of penalized coefficients can yield an improper posterior \citep{lang2004bayesian}. Although the prior on $\mathbf{B}$ remains partially improper---flat on $\textnormal{null}(\mathbf{P}\otimes\mathbf{I}_R)$---we prove in Section~\ref{sec:supp-propriety} that the resulting posterior distribution is proper under a single checkable condition on the observation pattern. 

\subsection{Derivation of the Conditional Posterior Distributions}
\label{sec:supp-conditionals}
We now derive the full conditional distributions used in the Gibbs sampler of Algorithm~1 of the main manuscript, given the completed data $\mathbf{X}$.

\medskip
\noindent\textbf{Full conditional of $\textnormal{Vec}(\mathbf{B})$.}\quad
Write $\mathbf{y}=\textnormal{Vec}(\mathbf{X})$ and $\mathbf{Z}=\mathbf{\Theta}^T\otimes\mathbf{V}\in\mathbb{R}^{NK\times RL}$. Combining the likelihood \eqref{eq:supp-lik} with the prior \eqref{eq:supp-priorB}, the log full conditional of $\textnormal{Vec}(\mathbf{B})$ equals, up to an additive constant,
\begin{equation*}
    -\tfrac{1}{2\sigma^2}\big(\mathbf{y}-\mathbf{Z}\,\textnormal{Vec}(\mathbf{B})\big)^T\big(\mathbf{y}-\mathbf{Z}\,\textnormal{Vec}(\mathbf{B})\big) - \tfrac{1}{2\sigma_B^2}\textnormal{Vec}(\mathbf{B})^T(\mathbf{P}\otimes\mathbf{I}_R)\textnormal{Vec}(\mathbf{B})\;.
\end{equation*}
Expanding the first term and collecting the quadratic and linear parts in $\textnormal{Vec}(\mathbf{B})$ gives
\begin{equation*}
    -\tfrac12\,\textnormal{Vec}(\mathbf{B})^T\underbrace{\Big(\tfrac{1}{\sigma^2}\mathbf{Z}^T\mathbf{Z}+\tfrac{1}{\sigma_B^2}(\mathbf{P}\otimes\mathbf{I}_R)\Big)}_{\displaystyle \boldsymbol{\Sigma}_B^{-1}}\textnormal{Vec}(\mathbf{B})\; +\; \tfrac{1}{\sigma^2}\textnormal{Vec}(\mathbf{B})^T\mathbf{Z}^T\mathbf{y}\; +\; \textnormal{const}\;.
\end{equation*}
This is the kernel of a Gaussian distribution; completing the square yields
\begin{equation*}
    \textnormal{Vec}(\mathbf{B})\mid\mathbf{V},\sigma^2,\sigma_B^2,\mathbf{X}\;\sim\; N(\boldsymbol{\mu}_B,\boldsymbol{\Sigma}_B)\;,
\end{equation*}
with
\begin{align*}
    \boldsymbol{\Sigma}_B &= \left(\tfrac{1}{\sigma^2}(\mathbf{\Theta}^T\otimes\mathbf{V})^T(\mathbf{\Theta}^T\otimes\mathbf{V}) + \tfrac{1}{\sigma_B^2}(\mathbf{P}\otimes\mathbf{I}_R)\right)^{-1}\;,\\
    \boldsymbol{\mu}_B &= \tfrac{1}{\sigma^2}\,\boldsymbol{\Sigma}_B\,(\mathbf{\Theta}^T\otimes\mathbf{V})^T\textnormal{Vec}(\mathbf{X})\;.
\end{align*}

\medskip
\noindent\textbf{Full conditional of $\mathbf{V}$.}\quad
Conditioning on $\mathbf{U}^T=\mathbf{B}\mathbf{\Theta}$ and $\sigma^2$, model \eqref{eq:supp-model} decouples across the rows of $\mathbf{V}$. Let $\boldsymbol{v}_{i\cdot}$ denote the $i$-th row of $\mathbf{V}$; transposing the $i$-th row of \eqref{eq:supp-model} gives
\begin{equation*}
    \mathbf{X}_{i\cdot}^T = \mathbf{U}\,\boldsymbol{v}_{i\cdot}^T + \mathbf{E}_{i\cdot}^T\;,\qquad \mathbf{E}_{i\cdot}^T\sim N(\mathbf{0},\sigma^2\mathbf{I}_K)\;,
\end{equation*}
a Gaussian linear model in $\boldsymbol{v}_{i\cdot}^T\in\mathbb{R}^R$ with design matrix $\mathbf{U}\in\mathbb{R}^{K\times R}$. Combining this likelihood with the Gaussian prior \eqref{eq:supp-pp-priorV} and completing the square gives, independently for $i=1,\dots,N$,
\begin{equation*}
    \boldsymbol{v}_{i\cdot}^T\mid\mathbf{U},\sigma^2,\mathbf{X}\ \sim\
    N\!\left(\tfrac{1}{\sigma^2}\boldsymbol{\Sigma}_v\mathbf{U}^T\mathbf{X}_{i\cdot}^T,\ \ \boldsymbol{\Sigma}_v\right)\;,
    \qquad
    \boldsymbol{\Sigma}_v=\Bigl(\tfrac{1}{\sigma^2}\mathbf{U}^T\mathbf{U}+\tfrac{1}{\sigma_V^2}\mathbf{I}_R\Bigr)^{-1}\;.
\end{equation*}
The covariance $\boldsymbol{\Sigma}_v$ does not depend on $i$ and need be formed only once per sweep. The Gaussian prior acts as a ridge, so $\boldsymbol{\Sigma}_v$ is positive definite even when $\mathbf{U}^T\mathbf{U}$ is singular or ill-conditioned.

\medskip
\noindent\textbf{Full conditional of $\sigma^2$.}\quad
Let $\widehat{\mathbf{X}}=\mathbf{V}\mathbf{U}^T=\mathbf{V}\mathbf{B}\mathbf{\Theta}$ be fixed by the current draws of $\mathbf{V}$ and $\mathbf{B}$, and note that $\|\textnormal{Vec}(\mathbf{X})-(\mathbf{\Theta}^T\otimes\mathbf{V})\textnormal{Vec}(\mathbf{B})\|^2=\|\mathbf{X}-\widehat{\mathbf{X}}\|_F^2$, where $\|\cdot\|_F$ is the Frobenius norm. As a function of $\sigma^2$, the likelihood \eqref{eq:supp-lik} is proportional to $(\sigma^2)^{-NK/2}\exp\!\big(-\|\mathbf{X}-\widehat{\mathbf{X}}\|_F^2/(2\sigma^2)\big)$; multiplying by the prior \eqref{eq:supp-prop-priorsig} gives the inverse-gamma kernel
\begin{equation*}
    \sigma^2\mid\mathbf{V},\mathbf{B},\mathbf{X}\;\sim\;\textnormal{IG}\!\left(\tfrac{NK}{2}+a_\sigma,\;\; \tfrac12\|\mathbf{X}-\widehat{\mathbf{X}}\|_F^2+b_\sigma\right)\;.
\end{equation*}

\medskip
\noindent\textbf{Full conditional of $\sigma_B^2$.}\quad
Among all terms, only the prior \eqref{eq:supp-priorB} depends on $\sigma_B^2$. Holding $\mathbf{B}$ fixed and multiplying by the prior \eqref{eq:supp-pp-priorsigB} gives
\begin{equation*}
    p(\sigma_B^2\mid\mathbf{B})\propto (\sigma_B^2)^{-qR/2-a_B-1}\exp\!\left(-\frac{\tfrac12\textnormal{Vec}(\mathbf{B})^T(\mathbf{P}\otimes\mathbf{I}_R)\textnormal{Vec}(\mathbf{B})+b_B}{\sigma_B^2}\right)\;,
\end{equation*}
which is again an inverse-gamma kernel,
\begin{equation}\label{eq:supp-pp-sigB}
    \sigma_B^2\mid\mathbf{B}\;\sim\;\textnormal{IG}\!\left(\frac{qR}{2}+a_B,\;\; \frac12\textnormal{Vec}(\mathbf{B})^T(\mathbf{P}\otimes\mathbf{I}_R)\textnormal{Vec}(\mathbf{B})+b_B\right)\;,
\end{equation}
with $q=L-2$. The shape constant $qR/2$ equals one half of $\textnormal{rank}(\mathbf{P}\otimes\mathbf{I}_R)$, consistent with the rank-based normalization of \eqref{eq:supp-priorB}; under the dimension-based kernel the same update has shape $RL/2+a_B$, corresponding to the equivalent hyperprior $\textnormal{IG}(a_B+mR/2,b_B)$ noted in Section~\ref{sec:supp-model-spec}.

Iterating these four updates constitutes one sweep of the Gibbs sampler in Algorithm~1 of the main manuscript.

\subsection{A Fixed--Random Reparameterization of the Smoothness Prior}
\label{sec:supp-proper}
The prior \eqref{eq:supp-priorB} is flat along $\textnormal{null}(\mathbf{P}\otimes\mathbf{I}_R)$. In this subsection we make the improper directions explicit by decomposing the spline coefficients into unpenalized \emph{fixed effects} and penalized \emph{random effects}; this is the spectral reparameterization employed for Bayesian functional regression in \citet[Section~2.2]{jiang2025tutorial}, and it underlies the mixed-model representation of penalized splines \citep{ruppert2003semiparametric}. We emphasize that the decomposition is a bijective linear change of variables and does not alter the model; it serves purely as a device for the propriety analysis of Supplementary Material Section~\ref{sec:supp-propriety}, and the sampler of Section~\ref{sec:supp-conditionals} operates on $\textnormal{Vec}(\mathbf{B})$ directly.

For the second-derivative penalty matrix $\mathbf{P}$,
\begin{equation*}
    \textnormal{null}(\mathbf{P})=\bigl\{\mathbf{b}\in\mathbb{R}^{L}:\ (\mathbf{b}^T\mathbf{\Theta})''(t)\equiv 0\bigr\}
\end{equation*}
is the set of coefficient vectors of affine functions, so that $m=2$ and $q=L-2$: the constant and linear components of each functional factor are left unpenalized.

Let $\mathbf{P}=\mathbf{Q}\mathbf{D}\mathbf{Q}^T$ be a spectral decomposition, where
$\mathbf{Q}=[\mathbf{Q}_0,\ \mathbf{Q}_+]\in\mathbb{R}^{L\times L}$ is orthogonal,
$\mathbf{D}=\textnormal{Block-Diagonal}(\mathbf{0}_m,\mathbf{D}_+)$, where
$\mathbf{D}_+=\textnormal{diag}(d_1,\ldots,d_q)$ contains the $q=L-m$ positive eigenvalues
of $\mathbf{P}$; the columns of $\mathbf{Q}_0\in\mathbb{R}^{L\times m}$ span
$\textnormal{null}(\mathbf{P})$ ($m=2$ for the second-order penalty), so that
$\mathbf{P}=\mathbf{Q}_+\mathbf{D}_+\mathbf{Q}_+^T$.
Following the reparameterization in Section~2.2 of \citet{jiang2025tutorial}, define, for
$r=1,\ldots,R$, $\boldsymbol{\gamma}_r=\mathbf{Q}_0^T\mathbf{B}_{r\cdot}\in\mathbb{R}^{m}$, and $\boldsymbol{\eta}_r=\mathbf{D}_+^{1/2}\mathbf{Q}_+^T\mathbf{B}_{r\cdot}\in\mathbb{R}^{q}$, then
\begin{equation}\label{eq:supp-pp-rowsplit}
\mathbf{B}_{r\cdot}=\mathbf{Q}_0\boldsymbol{\gamma}_r+\mathbf{Q}_+\mathbf{D}_+^{-1/2}\boldsymbol{\eta}_r\;,
\end{equation}
an invertible linear change of coordinates. Since
$\mathbf{P}=\mathbf{Q}_+\mathbf{D}_+\mathbf{Q}_+^T$ and $\mathbf{Q}_+^T\mathbf{Q}_0=\mathbf{0}$,
\begin{equation}\label{eq:supp-pp-quad}
    \mathbf{B}_{r\cdot}^T\mathbf{P}\,\mathbf{B}_{r\cdot}=\boldsymbol{\eta}_r^T\boldsymbol{\eta}_r\;,
    \qquad r=1,\ldots,R\;,
\end{equation}
that is, the penalty $\mathbf{P}$ on $\mathbf{B}_{r\cdot}$ is equivalent to the identity
penalty on $\boldsymbol{\eta}_r$ with smoothing parameter $\sigma_B^2$ and no penalty on
$\boldsymbol{\gamma}_r$ \citep{jiang2025tutorial}. Stacking the rows into
$\boldsymbol{\Gamma}\in\mathbb{R}^{R\times m}$ and $\mathbf{H}\in\mathbb{R}^{R\times q}$, with
$\boldsymbol{\gamma}=\textnormal{Vec}(\boldsymbol{\Gamma})$ and
$\boldsymbol{\eta}=\textnormal{Vec}(\mathbf{H})$, equation \eqref{eq:supp-pp-rowsplit} becomes
$\mathbf{B}=\boldsymbol{\Gamma}\mathbf{Q}_0^T+\mathbf{H}\mathbf{D}_+^{-1/2}\mathbf{Q}_+^T$, so that
\begin{equation}\label{eq:supp-pp-etaeta}
    \textnormal{Vec}(\mathbf{B})
    =(\mathbf{Q}_0\otimes\mathbf{I}_R)\,\boldsymbol{\gamma}
    +(\mathbf{Q}_+\mathbf{D}_+^{-1/2}\otimes\mathbf{I}_R)\,\boldsymbol{\eta}\;.
\end{equation}
Summed over $r$, equation \eqref{eq:supp-pp-quad} becomes
\begin{equation}
    \textnormal{Vec}(\mathbf{B})^T(\mathbf{P}\otimes\mathbf{I}_R)\textnormal{Vec}(\mathbf{B})
    =\boldsymbol{\eta}^T\boldsymbol{\eta}\;.
\end{equation}
The columns of $\mathbf{Q}_0\otimes\mathbf{I}_R$ span
$\textnormal{null}(\mathbf{P}\otimes\mathbf{I}_R)$, of dimension $mR=2R$.

Consequently, the smoothness prior \eqref{eq:supp-priorB} is \emph{exactly} the pair of
independent priors
\begin{equation}\label{eq:supp-pp-mixedprior}
    p(\boldsymbol{\gamma})\propto 1
    \qquad\text{and}\qquad
    \boldsymbol{\eta}\mid\sigma_B^2\ \sim\
    \textnormal{MVN}\!\left(\mathbf{0},\,\sigma_B^2\mathbf{I}_{qR}\right),
\end{equation}
$\boldsymbol{\gamma}$ collects the $mR=2R$ unpenalized \emph{fixed effects} and
$\boldsymbol{\eta}$ the $qR=R(L-2)$ penalized \emph{random effects}, and the factor
$(\sigma_B^2)^{-qR/2}$ in \eqref{eq:supp-priorB} is the normalizing constant of the Gaussian
law of $\boldsymbol{\eta}$; all impropriety induced by the rank deficiency of $\mathbf{P}$ is
confined to $\boldsymbol{\gamma}$. Because the change of coordinates
\eqref{eq:supp-pp-etaeta} is linear and invertible with a Jacobian that is constant in all
model parameters, the two parameterizations induce the same posterior on
$(\mathbf{V},\mathbf{B},\sigma^2,\sigma_B^2)$. The previous results are all analogous to those in \citet{jiang2025tutorial}.

\medskip
\noindent\textbf{Design decomposition.}\quad
Let $\mathbf{\Theta}_0=\mathbf{Q}_0^T\mathbf{\Theta}\in\mathbb{R}^{m\times K}$ and
$\mathbf{\Theta}_1=\mathbf{D}_+^{-1/2}\mathbf{Q}_+^T\mathbf{\Theta}\in\mathbb{R}^{q\times K}$.
For the second-derivative penalty, the rows of $\mathbf{\Theta}_0$ span the same space as the
constant and linear functions evaluated on the grid, namely
$\textnormal{span}\{(1,\dots,1),\,(t_1,\dots,t_K)\}$, since the spline basis reproduces linear
functions; hence $\mathbf{\Theta}_0$ has full row rank $m=2$ whenever the grid contains at
least two distinct time points. Substituting
$\mathbf{B}=\boldsymbol{\Gamma}\mathbf{Q}_0^T+\mathbf{H}\mathbf{D}_+^{-1/2}\mathbf{Q}_+^T$ into
\eqref{eq:supp-model} gives
$\mathbf{V}\mathbf{B}\mathbf{\Theta}=\mathbf{V}\boldsymbol{\Gamma}\mathbf{\Theta}_0+\mathbf{V}\mathbf{H}\mathbf{\Theta}_1$,
so that with $\mathbf{y}=\textnormal{Vec}(\mathbf{X})$ the vectorized model \eqref{eq:supp-vec}
becomes the Gaussian linear mixed model
\begin{equation}\label{eq:supp-pp-lmm}
    \mathbf{y}=\mathbf{Z}_0\boldsymbol{\gamma}+\mathbf{Z}_1\boldsymbol{\eta}+\textnormal{Vec}(\mathbf{E})\;,
\end{equation}
where $\mathbf{Z}_0=(\mathbf{\Theta}^T\otimes\mathbf{V})(\mathbf{Q}_0\otimes\mathbf{I}_R)=\mathbf{\Theta}_0^T\otimes\mathbf{V}$, and $ \mathbf{Z}_1=(\mathbf{\Theta}^T\otimes\mathbf{V})(\mathbf{Q}_+\mathbf{D}_+^{-1/2}\otimes\mathbf{I}_R)=\mathbf{\Theta}_1^T\otimes\mathbf{V}$
with $\mathbf{Z}_0\in\mathbb{R}^{NK\times mR}$ and $\mathbf{Z}_1\in\mathbb{R}^{NK\times qR}$.
The Kronecker structure yields the identities
\begin{equation}\label{eq:supp-pp-kron}
    \mathbf{Z}_0^T\mathbf{Z}_0=(\mathbf{\Theta}_0\mathbf{\Theta}_0^T)\otimes(\mathbf{V}^T\mathbf{V}),
    \qquad
    \mathbf{Z}_1^T\mathbf{Z}_1=(\mathbf{\Theta}_1\mathbf{\Theta}_1^T)\otimes(\mathbf{V}^T\mathbf{V}),
    \qquad
    \mathbf{Z}_0^T\mathbf{Z}_1=(\mathbf{\Theta}_0\mathbf{\Theta}_1^T)\otimes(\mathbf{V}^T\mathbf{V}).
\end{equation}
In particular, $\textnormal{rank}(\mathbf{Z}_0)=m\cdot\textnormal{rank}(\mathbf{V})$, so
$\mathbf{Z}_0$ has full column rank $mR$ if and only if, in addition, $\mathbf{V}$ has full
column rank $R$. For the multiway model of Section~4 of the main manuscript, the same
decomposition applies with $\mathbf{\Theta}^T\otimes(\mathbf{W}\odot\mathbf{V})$ in place of
$\mathbf{\Theta}^T\otimes\mathbf{V}$, so that
$\mathbf{Z}_0=\mathbf{\Theta}_0^T\otimes(\mathbf{W}\odot\mathbf{V})$ and
$\mathbf{Z}_1=\mathbf{\Theta}_1^T\otimes(\mathbf{W}\odot\mathbf{V})$.

\makeatletter
\@ifundefined{thm}{\newtheorem{thm}{Theorem}[section]}{}
\@ifundefined{lem}{\newtheorem{lem}[thm]{Lemma}}{}
\@ifundefined{assumption}{\newtheorem{assumption}[thm]{Assumption}}{}
\@ifundefined{rmk}{\newtheorem{rmk}[thm]{Remark}}{}
\@ifundefined{proof}{\newenvironment{proof}[1][Proof]{\par\noindent\textit{#1.}\ }{\hfill$\square$\par\medskip}}{}
\providecommand{\qedhere}{}
\makeatother

\section{Propriety of the Posterior Distribution}
\label{sec:supp-propriety}

This section proves that the posterior distribution of the single-level BAMIFun model is proper under the prior specification described above. Posterior propriety under partially improper priors has been established for related models. For example, \citet{sun2001propriety} proved posterior propriety for the linear mixed model with a Gaussian prior on the random-effect coefficients and a flat prior on the fixed-effect coefficients. \citet{lang2004bayesian} focused on the Bayesian model for penalized splines. However, unlike these works, which consider regression models with fixed and known design matrices, two features of our model prevent a direct appeal to existing propriety results. First, results in \citet{sun2001propriety} guarantee only that, for each fixed $\mathbf{V}$, the conditional marginal likelihood $m(\mathbf{X}_{\mathrm{o}}\mid\mathbf{V})$ of the observed data is finite. Because the design matrices $\mathbf{Z}_0$ and $\mathbf{Z}_1$ in \eqref{eq:supp-pp-lmm} depend on the random score matrix $\mathbf{V}$, finiteness for every $\mathbf{V}$ together with a proper prior on $\mathbf{V}$ does not by itself imply that $\int m(\mathbf{X}_{\mathrm{o}}\mid\mathbf{V})\, p(\mathbf{V})\,d\mathbf{V}<\infty$. Second, the data are observed with missingness, so propriety must be established for the observed-data posterior. In this section, we address both issues and prove that the posterior distribution of our Bayesian imputation model is proper.

\subsection{Observed-data Posterior}
The observation indicator $\mathcal{O}\in\{0,1\}^{N\times K}$ and the rank $R$ are fixed.  Let $n_{\mathrm{o}}=\#\{(i,k):\mathcal{O}_{ik}=1\}$ denote the number of observed entries and let $\mathbf{M}\in\{0,1\}^{n_{\mathrm{o}}\times NK}$ be the selection matrix whose rows are the distinct standard basis vectors of $\mathbb{R}^{NK}$ picking out the observed coordinates of $\textnormal{Vec}(\mathbf{X})$, so that the observed data vector is $\mathbf{X}_{\mathrm{o}}=\mathbf{M}\,\textnormal{Vec}(\mathbf{X}) \in\mathbb{R}^{n_{\mathrm{o}}}$.  Because the entries of $\mathbf{E}$ are i.i.d.\ Gaussian, the observed-data likelihood is exactly
\begin{equation}\label{eq:supp-prop-lik}
    \mathbf{X}_{\mathrm{o}}\mid\mathbf{V},\mathbf{B},\sigma^2
    \ \sim\
    N\bigl(\mathbf{M}(\mathbf{\Theta}^T\otimes\mathbf{V})
    \textnormal{Vec}(\mathbf{B}),\ \sigma^2\mathbf{I}_{n_{\mathrm{o}}}\bigr)\;,
\end{equation}
and, by the reparameterization and design decomposition of
Section~\ref{sec:supp-proper} (cf.\ \eqref{eq:supp-pp-lmm}), its mean decomposes as
$\mathbf{Z}_0^{\mathrm{o}}\boldsymbol{\gamma}
+\mathbf{Z}_1^{\mathrm{o}}\boldsymbol{\eta}$ with the observed-data
designs
\begin{equation*}
    \mathbf{Z}_0^{\mathrm{o}}
    =\mathbf{M}(\mathbf{\Theta}_0^T\otimes\mathbf{V})
    \in\mathbb{R}^{n_{\mathrm{o}}\times mR}\;,
    \qquad
    \mathbf{Z}_1^{\mathrm{o}}
    =\mathbf{M}(\mathbf{\Theta}_1^T\otimes\mathbf{V})
    \in\mathbb{R}^{n_{\mathrm{o}}\times qR}\;.
\end{equation*}
Posterior propriety refers to the finiteness of the observed-data marginal
\begin{equation}\label{eq:supp-prop-marginal}
    m(\mathbf{X}_{\mathrm{o}})
    =\int
    p\bigl(\mathbf{X}_{\mathrm{o}}\mid\mathbf{V},\mathbf{B},\sigma^2\bigr)\,
    p\bigl(\textnormal{Vec}(\mathbf{B})\mid\sigma_B^2\bigr)\,
    p(\mathbf{V})\,p(\sigma^2)\,p(\sigma_B^2)\,
    d\mathbf{B}\,d\mathbf{V}\,d\sigma^2\,d\sigma_B^2 \;,
\end{equation}
where $p(\textnormal{Vec}(\mathbf{B})\mid\sigma_B^2)$ is the smoothness prior, equivalently the pair
\eqref{eq:supp-pp-mixedprior}, and $p(\mathbf{V})$ is the Gaussian
prior \eqref{eq:supp-pp-priorV}.  Once
$m(\mathbf{X}_{\mathrm{o}})<\infty$, the data-augmented posterior of
$(\mathbf{V},\mathbf{B},\sigma^2,\sigma_B^2,
\{\mathbf{X}_{ik}:\mathcal{O}_{ik}=0\})$ targeted by the Gibbs sampler
is also proper, because the missing entries have a proper Gaussian
conditional distribution given the parameters.

We impose one weak condition on the observation pattern.

\begin{assumption}\label{ass:supp-pattern}
There exists a subset $\mathcal{S}\subseteq\{1,\dots,N\}$ with
$|\mathcal{S}|\geq R+2$ such that every subject $i\in\mathcal{S}$ is
observed at two or more distinct time points; that is,
$\#\{t_k:\mathcal{O}_{ik}=1\}\geq 2$.
\end{assumption}

Assumption~\ref{ass:supp-pattern} is weak and readily checkable: at least $R+2$
subjects must each be observed at two or more distinct time points,
and the subjects need not share any common observation times.  It
fails only when at most $R+1$ subjects possess two or more distinct
observation times. Thus, it holds for essentially any realistic design,
however sparse or irregular.  No condition on the observed values
$\mathbf{X}_{\mathrm{o}}$ themselves is required.

\subsection{Lemmas}
We first use the following results from \citet{sun2001propriety}:

\begin{lem}\label{lem:supp-condbound}
Fix $\mathbf{V}$ such that
$\mathbf{Z}_0^{\mathrm{o}}=\mathbf{Z}_0^{\mathrm{o}}(\mathbf{V})$ has
full column rank $mR$, and suppose $n_{\mathrm{o}}\geq mR$.  Then
\begin{equation}\label{eq:supp-prop-condbound}
\begin{aligned}
    m(\mathbf{X}_{\mathrm{o}}\mid\mathbf{V})
    &:=\int
    p(\mathbf{X}_{\mathrm{o}}\mid\mathbf{V},\boldsymbol{\gamma},
    \boldsymbol{\eta},\sigma^2)\,
    p(\boldsymbol{\eta}\mid\sigma_B^2)\,
    p(\sigma^2)\,p(\sigma_B^2)\,
    d\boldsymbol{\gamma}\,d\boldsymbol{\eta}\,d\sigma^2\,d\sigma_B^2 \\
    &\ \leq\ C_0\,
    \bigl|(\mathbf{Z}_0^{\mathrm{o}})^T\mathbf{Z}_0^{\mathrm{o}}\bigr|^{-1/2}\;,
    \qquad
    C_0=\frac{\Gamma\bigl((n_{\mathrm{o}}-mR)/2+a_\sigma\bigr)}
    {\Gamma(a_\sigma)\,(2\pi b_\sigma)^{(n_{\mathrm{o}}-mR)/2}} \;,
\end{aligned}
\end{equation}
where $C_0$ depends only on $n_{\mathrm{o}}$, $R$, $a_\sigma$ and
$b_\sigma$.
\end{lem}

\begin{proof}
Write
$\boldsymbol{\Pi}_0=\mathbf{I}_{n_{\mathrm{o}}}
-\mathbf{Z}_0^{\mathrm{o}}\{(\mathbf{Z}_0^{\mathrm{o}})^T
\mathbf{Z}_0^{\mathrm{o}}\}^{-1}(\mathbf{Z}_0^{\mathrm{o}})^T$,
the orthogonal projector onto the orthogonal complement of the column
space of $\mathbf{Z}_0^{\mathrm{o}}$.  All
integrands are nonnegative, so the order of integration below is
justified by Tonelli's theorem.

\emph{Step 1 (integrate $\boldsymbol{\gamma}$).}  Because
$p(\boldsymbol{\gamma})\propto 1$, for fixed
$(\boldsymbol{\eta},\sigma^2)$,
\begin{align*}
    &\int(2\pi\sigma^2)^{-n_{\mathrm{o}}/2}
    \exp\!\Bigl(-\tfrac{1}{2\sigma^2}
    \bigl\|\mathbf{X}_{\mathrm{o}}
    -\mathbf{Z}_0^{\mathrm{o}}\boldsymbol{\gamma}
    -\mathbf{Z}_1^{\mathrm{o}}\boldsymbol{\eta}\bigr\|^2\Bigr)
    d\boldsymbol{\gamma}\\
    &\qquad=(2\pi\sigma^2)^{-(n_{\mathrm{o}}-mR)/2}
    \bigl|(\mathbf{Z}_0^{\mathrm{o}})^T\mathbf{Z}_0^{\mathrm{o}}
    \bigr|^{-1/2}
    \exp\!\Bigl(-\tfrac{1}{2\sigma^2}
    (\mathbf{X}_{\mathrm{o}}-\mathbf{Z}_1^{\mathrm{o}}
    \boldsymbol{\eta})^T\boldsymbol{\Pi}_0
    (\mathbf{X}_{\mathrm{o}}-\mathbf{Z}_1^{\mathrm{o}}
    \boldsymbol{\eta})\Bigr)\;.
\end{align*}

\emph{Step 2 (integrate $\boldsymbol{\eta}$).}  Multiplying by the
$N(\mathbf{0},\sigma_B^2\mathbf{I}_{qR})$ density of
$\boldsymbol{\eta}$ from \eqref{eq:supp-pp-mixedprior} and completing
the square with precision
$\boldsymbol{\Omega}_\eta=\sigma^{-2}(\mathbf{Z}_1^{\mathrm{o}})^T\boldsymbol{\Pi}_0
\mathbf{Z}_1^{\mathrm{o}}+\sigma_B^{-2}\mathbf{I}_{qR}$ gives, after
integrating, the factor
\begin{equation*}
    \bigl|\sigma_B^2\boldsymbol{\Omega}_\eta\bigr|^{-1/2}
    =\Bigl|\tfrac{\sigma_B^2}{\sigma^2}
    (\mathbf{Z}_1^{\mathrm{o}})^T\boldsymbol{\Pi}_0\mathbf{Z}_1^{\mathrm{o}}
    +\mathbf{I}_{qR}\Bigr|^{-1/2}\leq 1\;,
\end{equation*}
multiplied by the exponential of a nonpositive quantity, which is at
most one.  Hence
\begin{equation*}
    p(\mathbf{X}_{\mathrm{o}}\mid\mathbf{V},\sigma^2,\sigma_B^2)
    \ \leq\
    (2\pi\sigma^2)^{-(n_{\mathrm{o}}-mR)/2}
    \bigl|(\mathbf{Z}_0^{\mathrm{o}})^T\mathbf{Z}_0^{\mathrm{o}}
    \bigr|^{-1/2}\;.
\end{equation*}

\emph{Step 3 (integrate $\sigma^2$ and $\sigma_B^2$).}  Against the
$\textnormal{IG}(a_\sigma,b_\sigma)$ density
$\{b_\sigma^{a_\sigma}/\Gamma(a_\sigma)\}
(\sigma^2)^{-a_\sigma-1}\exp(-b_\sigma/\sigma^2)$,
\begin{equation*}
    \int_0^\infty(\sigma^2)^{-(n_{\mathrm{o}}-mR)/2}\,
    \frac{b_\sigma^{a_\sigma}}{\Gamma(a_\sigma)}\,
    (\sigma^2)^{-a_\sigma-1}
    \exp\!\bigl(-b_\sigma/\sigma^2\bigr)\,d\sigma^2
    =\frac{\Gamma\bigl((n_{\mathrm{o}}-mR)/2+a_\sigma\bigr)}
    {\Gamma(a_\sigma)}\,b_\sigma^{-(n_{\mathrm{o}}-mR)/2}\;,
\end{equation*}
which is finite because $a_\sigma,b_\sigma>0$ and
$n_{\mathrm{o}}\geq mR$; the proper
$\textnormal{IG}(a_B,b_B)$ prior for $\sigma_B^2$ integrates to one.
Collecting constants yields \eqref{eq:supp-prop-condbound}.
\end{proof}

Steps 1--3 parallel the closed-form integration in Theorem~1 of
\citet{sun2001propriety}: conditional on $\mathbf{V}$, the model
\eqref{eq:supp-pp-lmm} has flat-prior fixed effects and proper priors
on all variance components, so finiteness of
$m(\mathbf{X}_{\mathrm{o}}\mid\mathbf{V})$ for a fixed full-rank
design is precisely their result.  The display
\eqref{eq:supp-prop-condbound} additionally records the explicit
dependence of the bound on $\mathbf{Z}_0^{\mathrm{o}}(\mathbf{V})$,
which is what permits integration over the random design in the next
step. Integrating over the score matrix gives us the following results:

\begin{lem}\label{lem:supp-vmoment}
Under Assumption~\ref{ass:supp-pattern} and the Gaussian prior
\eqref{eq:supp-pp-priorV} on $\mathbf{V}$:
\begin{itemize}
    \item[(i)] $\mathbf{Z}_0^{\mathrm{o}}(\mathbf{V})$ has full column
    rank $mR$ for almost every $\mathbf{V}$;
    \item[(ii)] $\mathbb{E}\bigl[\,\bigl|(\mathbf{Z}_0^{\mathrm{o}})^T
    \mathbf{Z}_0^{\mathrm{o}}\bigr|^{-1/2}\bigr]<\infty$.
\end{itemize}
\end{lem}

\begin{proof}
The rows of $\mathbf{Z}_0^{\mathrm{o}}$ are indexed by the observed
pairs $(i,k)$ and are equal to
$(\mathbf{\Theta}_0)_{\cdot k}^T\otimes\boldsymbol{v}_{i\cdot}$, where
$(\mathbf{\Theta}_0)_{\cdot k}\in\mathbb{R}^{m}$ denotes the $k$-th
column of $\mathbf{\Theta}_0$.  Writing
$\mathbf{O}_i=\textnormal{diag}(\mathcal{O}_{i1},\dots,\mathcal{O}_{iK})$,
the matrix $(\mathbf{Z}_0^{\mathrm{o}})^T\mathbf{Z}_0^{\mathrm{o}}$ therefore decomposes subject-wise as
\begin{equation}\label{eq:supp-prop-gram}
    (\mathbf{Z}_0^{\mathrm{o}})^T\mathbf{Z}_0^{\mathrm{o}}
    =\sum_{i=1}^{N}\mathbf{A}_i\otimes
    \boldsymbol{v}_{i\cdot}^T\boldsymbol{v}_{i\cdot}\;,
\end{equation}
where 
\begin{equation*}
    \mathbf{A}_i:=\sum_{k:\,\mathcal{O}_{ik}=1}
    (\mathbf{\Theta}_0)_{\cdot k}(\mathbf{\Theta}_0)_{\cdot k}^T
    =\mathbf{\Theta}_0\mathbf{O}_i\mathbf{\Theta}_0^T = \big(\mathbf{\Theta}_0\mathbf{O}_i\big)\big(\mathbf{\Theta}_0\mathbf{O}_i\big)^T
    \in\mathbb{R}^{m\times m}
\end{equation*}
The last equation is because $\mathbf{O}_i$ is idempotent with $\mathbf{O}_i\mathbf{O}_i=\mathbf{O}_i$. Therefore, in order to show that $\mathbf{A}_i$ is positive definite, we only need to show that $\mathbf{\Theta}_0\mathbf{O}_i$ is of full row rank $m=2$.

Here, $\mathbf{\Theta}_0$ is an $m \times K$ matrix of spline basis that corresponds to the fixed effect (null space of the penalty matrix). The two rows of
\(\mathbf\Theta_0\) form a basis for
\(\operatorname{span}\{1,t\}\) evaluated on the observation grid.
Consequently, there exists a nonsingular matrix
\(\mathbf C\in\mathbb R^{2\times2}\) such that
\begin{equation}
    \boldsymbol\Theta_0
=
\mathbf C
\begin{pmatrix}
1& ...& 1\\ t_1& ...& t_K
\end{pmatrix}\;.
\end{equation}
For subject $i$ who has two distinct observation times \(t_{k_1}\ne t_{k_2}\), denote $\mathbf{B}$ as the submatrix of $\mathbf{\Theta}_0\mathbf{O}_i$ corresponds to observation times \(t_{k_1}\ne t_{k_2}\),
\[
\mathbf{B}_i
=
\mathbf C
\begin{pmatrix}
1 & 1\\
t_{k_1} & t_{k_2}
\end{pmatrix}\;,
\]
and hence
\[
\det\!\left[
\mathbf{B}_i
\right]
=
\det(\mathbf C)(t_{k_2}-t_{k_1})\ne0\;,
\]
which makes $\mathbf{B}_i$ of full rank. Thus, $\mathbf{\Theta}_0\mathbf{O}_i$ is also of full rank (since its submatrix $\mathbf{B}_i$ is of rank $m=2$) for any subject with two or more observations and \(\mathbf A_i\) is positive definite: \(\mathbf A_i\succ0\). 
Hence $\lambda_i:=\lambda_{\min}(\mathbf{A}_i)>0$, and by Assumption~\ref{ass:supp-pattern}, $\lambda:=\min_{i\in\mathcal{S}}\lambda_i>0$.



Discarding the positive semidefinite terms of
\eqref{eq:supp-prop-gram} with $i\notin\mathcal{S}$ and using
$\mathbf{A}_i\succeq\lambda_i\mathbf{I}_m\succeq\lambda\mathbf{I}_m$
for $i\in\mathcal{S}$, due to the fact that \(\mathbf{A}\otimes\mathbf{C}
\succeq
\mathbf{B}\otimes\mathbf{C},\) whenever \(\mathbf{A}\succeq\mathbf{B}\) and \(\mathbf{C}\succeq\mathbf{0}\), we have 
\begin{equation}\label{eq:supp-prop-minorize}
    (\mathbf{Z}_0^{\mathrm{o}})^T\mathbf{Z}_0^{\mathrm{o}}
    \ \succeq\
    \lambda\,\mathbf{I}_m\otimes
    \Bigl(\sum_{i\in\mathcal{S}}
    \boldsymbol{v}_{i\cdot}^T\boldsymbol{v}_{i\cdot}\Bigr)
    =\lambda\,\mathbf{I}_m\otimes
    (\mathbf{V}_{\mathcal{S}}^T\mathbf{V}_{\mathcal{S}})\;,
\end{equation}
where $\mathbf{V}_{\mathcal{S}}\in\mathbb{R}^{|\mathcal{S}|\times R}$
collects the rows of $\mathbf{V}$ with $i\in\mathcal{S}$. 
Since \(|\mathbf{A}|\geq|\mathbf{B}|\) whenever \(\mathbf{A}\succeq\mathbf{B}\succeq\mathbf{0},\) and $m=2$, we have,
\begin{equation*}
    \bigl|(\mathbf{Z}_0^{\mathrm{o}})^T\mathbf{Z}_0^{\mathrm{o}}\bigr|
    \ \geq\
    \lambda^{2R}\,
    \bigl|\mathbf{V}_{\mathcal{S}}^T\mathbf{V}_{\mathcal{S}}\bigr|^{2}\;.
\end{equation*}
The entries of $\mathbf{V}_{\mathcal{S}}$ are i.i.d.\
$N(0,\sigma_V^2)$ under \eqref{eq:supp-pp-priorV}, so
$\mathbf{W}:=\mathbf{V}_{\mathcal{S}}^T\mathbf{V}_{\mathcal{S}}$
follows the Wishart distribution
$W_R(|\mathcal{S}|,\sigma_V^2\mathbf{I}_R)$ and is almost surely
nonsingular because $|\mathcal{S}|\geq R$; with the preceding display
this proves (i).  

By the Bartlett decomposition
\citep{muirhead2009aspects},
$|\mathbf{W}|\overset{d}{=}\sigma_V^{2R}\prod_{r=1}^{R}
\chi^2_{|\mathcal{S}|-r+1}$ with independent chi-squared factors, and
$\mathbb{E}[1/\chi^2_\nu]=1/(\nu-2)$ for $\nu>2$. On the almost surely event \(\{|\mathbf{W}|>0\}\), we have 
\begin{equation*}
    \bigl|(\mathbf{Z}_0^{\mathrm{o}})^T\mathbf{Z}_0^{\mathrm{o}}\bigr|^{-1/2}
    \ \leq\
    \lambda^{-R}\,
    \bigl|\mathbf{W}\bigr|^{-1}\;.
\end{equation*}
Since
$|\mathcal{S}|\geq R+2$, every degree of freedom $r$ satisfies
$|\mathcal{S}|-r+1\geq|\mathcal{S}|-R+1\geq 3$, thus
\begin{equation*}
    \mathbb{E}\bigl[\,\bigl|(\mathbf{Z}_0^{\mathrm{o}})^T
    \mathbf{Z}_0^{\mathrm{o}}\bigr|^{-1/2}\bigr]
    \ \leq\
    \lambda^{-R}\,\mathbb{E}\bigl[\,|\mathbf{W}|^{-1}\bigr]
    =\lambda^{-R}\,\sigma_V^{-2R}
    \prod_{r=1}^{R}\frac{1}{|\mathcal{S}|-r-1}
    \ <\ \infty\;. \qedhere
\end{equation*}
\end{proof}

\subsection{Main Result}

\begin{thm}\label{thm:supp-propriety}
Consider the single-level model \eqref{eq:supp-model} with the
rank-normalized smoothness prior (equivalently
\eqref{eq:supp-pp-mixedprior}), the Gaussian prior
\eqref{eq:supp-pp-priorV} on $\mathbf{V}$, the inverse-gamma priors
\eqref{eq:supp-pp-priorsigB} on $\sigma_B^2$ and
\eqref{eq:supp-prop-priorsig} on $\sigma^2$, where
$a_\sigma,b_\sigma,a_B,b_B>0$ and $\sigma_V^2>0$.  Under
Assumption~\ref{ass:supp-pattern}, the posterior distribution of
$(\mathbf{V},\mathbf{B},\sigma^2,\sigma_B^2)$ given the observed data
$\mathbf{X}_{\mathrm{o}}$ is proper, i.e.,
$m(\mathbf{X}_{\mathrm{o}})<\infty$ in
\eqref{eq:supp-prop-marginal}.
\end{thm}

\begin{proof}
As shown in
Section~\ref{sec:supp-proper}, the map
$\textnormal{Vec}(\mathbf{B})\mapsto
(\boldsymbol{\gamma},\boldsymbol{\eta})$ is a linear bijection with
constant Jacobian, so integrating over $\mathbf{B}$ in
\eqref{eq:supp-prop-marginal} is equivalent to integrating over
$(\boldsymbol{\gamma},\boldsymbol{\eta})$ with the priors
\eqref{eq:supp-pp-mixedprior}.  Assumption~\ref{ass:supp-pattern}
gives $n_{\mathrm{o}}\geq 2|\mathcal{S}|\geq 2R+4>mR$, since each
subject in $\mathcal{S}$ contributes at least two observed entries.
On the almost-sure event of Lemma~\ref{lem:supp-vmoment}(i),
Lemma~\ref{lem:supp-condbound} therefore applies, and
\begin{equation*}
    m(\mathbf{X}_{\mathrm{o}})
    =\int m(\mathbf{X}_{\mathrm{o}}\mid\mathbf{V})\,
    p(\mathbf{V})\,d\mathbf{V}
    \ \leq\
    C_0\,\mathbb{E}\bigl[\,\bigl|(\mathbf{Z}_0^{\mathrm{o}})^T
    \mathbf{Z}_0^{\mathrm{o}}\bigr|^{-1/2}\bigr]
    \ <\ \infty
\end{equation*}
by Lemma~\ref{lem:supp-vmoment}(ii).
\end{proof}

In particular, the default hyperparameter choices of
Section~\ref{sec:supp-model-spec} ($\sigma_V^2=10^6$ and
$a_\sigma=b_\sigma=a_B=b_B=0.01$) satisfy the conditions of
Theorem~\ref{thm:supp-propriety}, so the posterior targeted by
Algorithm~1 of the main manuscript is proper whenever
Assumption~\ref{ass:supp-pattern} holds.

\section{Cross-Validation Algorithm for the Rank \(R\), With the Smoothing Parameter Sampled Through MCMC}
Algorithm \ref{alg2} presents the {\color{black}cross-validation algorithm for selecting the rank $R$}, where the smoothing parameter $\sigma_B^2$ is sampled through MCMC.
The cross-validation algorithm only requires the input of frequentist FPCA results as the initial values and the parameter set for cross-validation. For the default parameter grid, we propose to search across $R\in \{R_{\textnormal{freq}}-2, R_{\textnormal{freq}}-1, R_{\textnormal{freq}}\}$. Users may change the parameter grid based on their specific data structure. The cross-validation randomly selects $40\%$ of the observed data as the validation set and the rest of the data as the training set. The Bayesian functional imputation algorithm in Algorithm~1 of the main manuscript is then conducted using the training data. Parameters are selected to minimize the imputation mean squared error over the elements within the validation set. 

\begin{algorithm}
	\caption{Cross-validation algorithm for rank $R$ with $\sigma_B^2$ sampled via MCMC} 
 \label{alg2}
	\begin{algorithmic}[1]
        \State \textbf{Input}: observed functional data $\{\mathbf{X}, \mathcal{O}=1\}$.
        \State Run the frequentist face.sparse() algorithm to $\{\mathbf{X}, \mathcal{O}=1\}$ with percent variance explained $\textnormal{pve} = 0.99$ and get the estimated eigenfunctions $\hat{\mathbf{U}}_{\textnormal{init}}$, scores $\hat{\mathbf{V}}_{\textnormal{init}}$, and number of rank $R_{\textnormal{freq}}$.
        \State Set the parameter grid setting for cross-validation: $R\in \{R_{\textnormal{freq}}-2, R_{\textnormal{freq}}-1, R_{\textnormal{freq}}\}$.
        \State Randomly select $40\%$ of the observed data points as validation set $\mathcal{V}=1$ and the rest of the observed data as training set $\mathcal{V}=0$.
        \For {$c=1,...,C$-th cross-validation setting}
            \State Run the algorithm in Algorithm~1 of the main manuscript with settings $R=R[c]$ and initial value $\mathbf{U}=\hat{\mathbf{U}}_{\textnormal{init}}[c], \mathbf{V}=\hat{\mathbf{V}}_{\textnormal{init}}[c]$ where $R[c]$ are the $c$-th value of the parameter grid, $\hat{\mathbf{U}}_{\textnormal{init}}[c]$ is the first $R[c]$ eigenfunctions from $\hat{\mathbf{U}}_{\textnormal{init}}$ and $\hat{\mathbf{V}}_{\textnormal{init}}[c]$ similarly.
            \State Calculate the imputation mean squared error over the validation set $\mathcal{V}=1$.
        \EndFor	
        \State Select the parameter $R$ which has the smallest MSE over the validation set. 
        \State Run the algorithm in Algorithm~1 of the main manuscript with the selected parameter. 
        \State \textbf{Output}: Bayesian posterior draws of the imputed functional data $\{\hat{{\mathbf{X}}}^{1},...,\hat{{\mathbf{X}}}^S\}$.
	\end{algorithmic} 
\end{algorithm}

\section{Bayesian Multiple Imputation Algorithm, With the Smoothing Parameter Selected by Cross-Validation}\label{sec:supp-cv-sigmaB}

\begin{algorithm}
	\caption{{\color{black}Bayesian multiple imputation for single-level functional data with prespecified $\sigma_B^2$ and $R$}} 
 \label{alg1.supp}
	\begin{algorithmic}[1]
        \State \textbf{Input}: observed functional data $\{\mathbf{X}, \mathcal{O}=1\}$; number of principal components (npc) $R$; smoothing parameter $\sigma^2_B$; spline basis matrix $\mathbf{\Theta}$ and the penalty matrix $\mathbf{P}$.
        \State Run the frequentist face.sparse() algorithm to $\{\mathbf{X}, \mathcal{O}=1\}$ with $\textnormal{npc} = R$ and get the estimated eigenfunctions $\hat{\mathbf{U}}_{\textnormal{init}}$ and scores $\hat{\mathbf{V}}_{\textnormal{init}}$. 
        \State Set the initial value of $\mathbf{U}$ as $\hat{\mathbf{U}}_{\textnormal{init}}$ and $\mathbf{V}$ as $\hat{\mathbf{V}}_{\textnormal{init}}$. Impute the missing elements of $\{\mathbf{X}, \mathcal{O}=0\}$ using $\mathbf{V}\mathbf{U}^T$. Set the initial value of $\sigma^2$ equal to $1$.

        \For {$s=1,...,S$-th MCMC iteration}
            \State Sample $\hat{\mathbf{V}}^s$, $\hat{\mathbf{B}}^s$, and $(\hat{\sigma}^2)^{s}$ with the conditional posterior distributions.
            \State Calculate $(\hat{\mathbf{U}}^s)^T = \hat{\mathbf{B}}^s\mathbf{\Theta}$.
            \State Calculate the underlying structure of ${\mathbf{X}}$ as $\Tilde{{\mathbf{X}}}^s = \hat{\mathbf{V}}^s(\hat{\mathbf{U}}^s)^T$.
            \State Impute the missing value $\{\hat{{\mathbf{X}}}^s, \mathcal{O}=0\}$ following the Gaussian distribution with mean $\Tilde{{\mathbf{X}}}^s$ and variance $(\hat{\sigma}^2)^{s}$.
            \State Rescale $\hat{\mathbf{U}}^s$ and $\hat{\mathbf{V}}^s$ such that each eigenfunction has unit norm.
        \EndFor		
        \State \textbf{Output}: Bayesian posterior draws of the imputed functional data $\{\hat{{\mathbf{X}}}^{1},...,\hat{{\mathbf{X}}}^S\}$.

	\end{algorithmic} 
\end{algorithm}

Algorithm \ref{alg1.supp} presents a BAMIFun algorithm which, instead of sampling the smoothing parameter $\sigma_B^2$, fixes the value of $\sigma_B^2$ as prespecified. 
Since Algorithm \ref{alg1.supp} requires prespecifying the smoothing parameter $\sigma_B^2$ and the number of principal components $R$, we propose an alternative algorithm for users to select the smoothing parameter $\sigma_B^2$ through cross-validation. The algorithm with cross-validation is described in Algorithm \ref{alg2.supp}. Algorithm \ref{alg2.supp} only requires the input of frequentist FPCA results as the initial values and the parameter grid for cross-validation. For the default parameter grid, we propose to search across $R\in \{R_{\textnormal{freq}}-2, R_{\textnormal{freq}}-1, R_{\textnormal{freq}}\}$ and $\sigma^2_B\in \{0.2, 1,10,100\}$. Users may change the parameter grid based on their specific functional data. The cross-validation randomly selects $40\%$ of the observed data as a validation set and the rest of the data as the training set. The Bayesian functional imputation Algorithm \ref{alg1.supp} will be conducted using the training data with each combination of the parameters. The algorithm then selects the parameters that minimize the imputation mean squared error (MSE) on the validation set. 

\begin{algorithm}
	\caption{{\color{black}Cross-validation algorithm for rank $R$ and $\sigma_B^2$}} 
 \label{alg2.supp}
	\begin{algorithmic}[1]
        \State \textbf{Input}: observed functional data $\{\mathbf{X}, \mathcal{O}=1\}$.
        \State Run the frequentist face.sparse() algorithm to $\{\mathbf{X}, \mathcal{O}=1\}$ with percent variance explained $\textnormal{pve} = 0.99$ and get the estimated eigenfunctions $\hat{\mathbf{U}}_{\textnormal{init}}$, scores $\hat{\mathbf{V}}_{\textnormal{init}}$, and number of principal components $R_{\textnormal{freq}}$.
        \State Set the parameters grid setting for cross-validation: $R\in \{R_{\textnormal{freq}}-2, R_{\textnormal{freq}}-1, R_{\textnormal{freq}}\}$ and $\sigma^2_B\in \{0.2, 1,10,100\}$.
        \State Randomly select $40\%$ of the observed data points as validation set $\mathcal{V}=1$ and the rest of the observed data as training set $\mathcal{V}=0$.
        \For {$c=1,...,C$-th cross-validation setting}
            \State Run Algorithm \ref{alg1.supp} with settings $R=R[c]$ and $\sigma_B^2=\sigma_B^2[c]$ and initial value $\mathbf{U}=\hat{\mathbf{U}}_{\textnormal{init}}[c], \mathbf{V}=\hat{\mathbf{V}}_{\textnormal{init}}[c]$ where $R[c]$ and $\sigma_B^2[c]$ are the $c$-th value of the expanded parameter grid, $\hat{\mathbf{U}}_{\textnormal{init}}[c]$ is the first $R[c]$ eigenfunctions from $\hat{\mathbf{U}}_{\textnormal{init}}$ and $\hat{\mathbf{V}}_{\textnormal{init}}[c]$ similarly.
            \State Calculate the imputation mean squared error over the validation set $\mathcal{V}=1$.
        \EndFor	
        \State Select the parameter $R$ and $\sigma^2_B$ with the smallest MSE over the validation set. 
        \State Run Algorithm \ref{alg1.supp} with the selected parameter. 
        \State \textbf{Output}: Bayesian posterior draws of the imputed functional data $\{\hat{{\mathbf{X}}}^{1},...,\hat{{\mathbf{X}}}^S\}$.
	\end{algorithmic} 
\end{algorithm}

\section{Bayesian Multiple Imputation Algorithm for Multiway Functional Data}
\label{sec:supp-multiway}

\subsection{Difference Between the FTSVD and the MFPCA Models}
The FTSVD model directly specifies a low-rank decomposition of the entire functional tensor (e.g., a CP-type factorization), in which the latent structure is represented through multiplicatively coupled factors across the discrete modes (e.g., subject/feature/visit) and a smooth functional factor along the time domain. In contrast, the MFPCA framework treats the data as a collection of clustered/hierarchical curves ${X_{ij}(t)}$ and characterizes the dependence through an additive functional mixed-effects representation, decomposing between- and within-cluster variability via separate Karhunen–Loève expansions with level-specific eigenfunctions and scores. Thus, the FTSVD model is most suitable when the data are naturally a tensor, e.g., many subjects observed on a common time interval for multiple features/visits, so that each curve can be viewed as one slice of a coherent $N\times J\times K$ array and the dependence is well captured by a shared low-rank multilinear structure across modes. When the data are primarily clustered functional observations with irregular/sparse time points and strong additive within-cluster subject-specific deviations (random intercept/slope-type behavior), FTSVD is less appropriate because it does not directly separate between- and within-cluster variability or provide level-specific covariance/eigenfunction interpretations as in MFPCA. 
\subsection{Gibbs Sampler With Full Data}
\label{sec:supp-multiway-gibbs}
We now provide a Gibbs sampler algorithm, whose full conditional posterior distributions are as follows:

\begin{itemize}
    \item Conditional on $\mathbf{U}^T=\mathbf{B\Theta}$, $\mathbf{W}$, $\sigma^2$, and $\sigma_B^2$, for $i=1,\ldots,N$, draw the $i$-th row $\boldsymbol{v}_{i\cdot}$ of $\mathbf{V}$ from
    \begin{equation*}
       \boldsymbol{v}_{i\cdot}^T \sim N\!\left(\tfrac{1}{\sigma^2}\boldsymbol{\Sigma}_{v}(\mathbf{U}\odot\mathbf{W})^T[\mathcal{X}_{(1)}]_{i}^T,\;\; \boldsymbol{\Sigma}_{v}\right)\;,\qquad \boldsymbol{\Sigma}_{v}=\Bigl(\tfrac{1}{\sigma^2}(\mathbf{U}\odot\mathbf{W})^T(\mathbf{U}\odot\mathbf{W})+\tfrac{1}{\sigma_V^2}\mathbf{I}_R\Bigr)^{-1}\;,
    \end{equation*} 
where $[\mathcal{X}_{(1)}]_{i}$ is the $i$-th row of the matrix $\mathcal{X}_{(1)}$. 
    \item Conditional on $\mathbf{U}^T=\mathbf{B\Theta}$, $\mathbf{V}$, $\sigma^2$, and $\sigma_B^2$, for $j=1,\ldots,J$, draw the $j$-th row $\boldsymbol{w}_{j\cdot}$ of $\mathbf{W}$ from
    \begin{equation*}
       \boldsymbol{w}_{j\cdot}^T \sim N\!\left(\tfrac{1}{\sigma^2}\boldsymbol{\Sigma}_{w}(\mathbf{U}\odot\mathbf{V})^T[\mathcal{X}_{(2)}]_{j}^T,\;\; \boldsymbol{\Sigma}_{w}\right)\;,\qquad \boldsymbol{\Sigma}_{w}=\Bigl(\tfrac{1}{\sigma^2}(\mathbf{U}\odot\mathbf{V})^T(\mathbf{U}\odot\mathbf{V})+\tfrac{1}{\sigma_W^2}\mathbf{I}_R\Bigr)^{-1}\;,
    \end{equation*} 
where $[\mathcal{X}_{(2)}]_{j}$ is the $j$-th row of the matrix $\mathcal{X}_{(2)}$. 
    \item Conditional on $\mathbf{V}$, $\mathbf{W}$, $\sigma^2$, and $\sigma_B^2$, draw $\textnormal{Vec}(\mathbf{B})$ following $\textnormal{Vec}(\mathbf{B}) \sim N(\boldsymbol{\mu}_{B}, \boldsymbol{\Sigma}_{B})$ where
    \begin{equation*}
                \boldsymbol{\Sigma}_B = \bigg(\frac{1}{\sigma^2}(\mathbf{\Theta}^T \otimes (\mathbf{W}\odot\mathbf{V}))^T(\mathbf{\Theta}^T \otimes (\mathbf{W}\odot\mathbf{V}))+\frac{1}{\sigma_B^2}{\color{black}\mathbf{P}\otimes\mathbf{I}_R}\bigg)^{-1}\;,
    \end{equation*}
    \begin{equation*}
                \boldsymbol{\mu}_B = \frac{1}{\sigma^2} \boldsymbol{\Sigma}_B (\mathbf{\Theta}^T \otimes (\mathbf{W}\odot\mathbf{V}))^T\textnormal{Vec}(\mathcal{X}_{(3)}^T)\;.
    \end{equation*}
    \item Given $\mathbf{V}$, $\mathbf{U}^T = \mathbf{B} \mathbf{\Theta}$, $\mathbf{W}$, and $\sigma_B^2$, the distribution of $\sigma^2$ follows:
    \begin{equation*}
        \sigma^2\sim \textnormal{IG}\left(\tfrac{NJK}{2}+a_\sigma,\;\; \tfrac12||\mathcal{X}-\hat{\mathcal{X}}||^2_{F}+b_\sigma\right)\;,
    \end{equation*}
    where $\hat{\mathcal{X}}=\llbracket \mathbf{V}, \mathbf{W}, \mathbf{U}\rrbracket$, and $||\cdot||_F$ is the Frobenius norm.
    \item Given $\mathbf{V}$, $\mathbf{U}^T = \mathbf{B} \mathbf{\Theta}$, $\mathbf{W}$, and $\sigma^2$, the distribution of $\sigma_B^2$ follows:
    \begin{equation*}
        \sigma_B^2\sim \textnormal{IG}\left(\tfrac{qR}{2}+a_B,\;\; \tfrac12(\textnormal{Vec}(\mathbf{B}))^T(\mathbf{P}\otimes\mathbf{I}_R)\textnormal{Vec}(\mathbf{B})+b_B\right)\;,\qquad q=L-2\;.
    \end{equation*}
\end{itemize}

\subsection{Bayesian Multiple Imputation Algorithm}

\begin{algorithm}
	\caption{Bayesian multiple imputation for multiway functional data} 
 \label{alg3}
	\begin{algorithmic}[1]
        \State \textbf{Input}: observed functional data $\{\mathcal{X}, \mathcal{O}=1\}$; rank $R$; spline basis matrix $\mathbf{\Theta}$ and the penalty matrix $\mathbf{P}$.
        \State Set the initial value of $\mathbf{U}$ as $(\hat{\mathbf{B}}_{\textnormal{init}}\mathbf{\Theta})^T$, $\mathbf{V}$ as $\hat{\mathbf{V}}_{\textnormal{init}}$, $\mathbf{W}$ as $\hat{\mathbf{W}}_{\textnormal{init}}$, $\sigma^2=1$, and $\sigma_B^2=1$. Impute the missing elements of $\{\mathcal{X}, \mathcal{O}=0\}$ using $\llbracket \hat{\mathbf{V}}_{\textnormal{init}}, \hat{\mathbf{W}}_{\textnormal{init}}, \hat{\mathbf{U}}_{\textnormal{init}}\rrbracket$.

        \For {$s=1,...,S$-th MCMC iteration}
            \State Draw $\hat{\mathbf{V}}^s$, $\hat{\mathbf{B}}^s$, $\hat{\mathbf{W}}^s$ one by one, following their full conditional distributions described before.
            \State Calculate $(\hat{\mathbf{U}}^s)^T = \hat{\mathbf{B}}^s\mathbf{\Theta}$.
            \State Draw $(\hat{\sigma}^2)^{s}$ and $(\hat{\sigma}_B^2)^{s}$ following their full conditional distributions.
            \State Calculate the underlying structure of ${\mathcal{X}}$ as $\Tilde{{\mathcal{X}}}^s = \llbracket \hat{\mathbf{V}}^s, \hat{\mathbf{W}}^s, \hat{\mathbf{U}}^s\rrbracket$.
            \State Impute the missing value $\{\hat{{\mathcal{X}}}^s, \mathcal{O}=0\}$ following the Gaussian distribution with mean $\Tilde{{\mathcal{X}}}^s$ and variance $(\hat{\sigma}^2)^{s}$.
            \State Rescale each column of $\hat{\mathbf{U}}^s$, $\hat{\mathbf{V}}^s$, and $\hat{\mathbf{W}}^s$ such that each column has the same norm.
        \EndFor		
        \State \textbf{Output}: Bayesian posterior draws of the imputed functional data $\{\hat{{\mathcal{X}}}^{1},...,\hat{{\mathcal{X}}}^S\}$.
	\end{algorithmic} 
\end{algorithm}

The Bayesian multiple imputation algorithm for multiway functional data is described in Algorithm \ref{alg3}. Similar to the algorithm for single-level functional data, Algorithm \ref{alg3} requires pre-specifying the rank $R$, spline basis matrix $\mathbf{\Theta}$ and the corresponding penalty matrix $\mathbf{P}$. The initial values of $\hat{\mathbf{V}}_{\textnormal{init}}$, $\hat{\mathbf{W}}_{\textnormal{init}}$, and $\hat{\mathbf{B}}_{\textnormal{init}}$ can be randomly drawn from a standard Gaussian distribution. In the rescale step, for each $r=1,...,R$, we rescale 
\begin{equation*}
    \hat{\mathbf{U}}^s_{\cdot r} = \hat{\mathbf{U}}^s_{\cdot r} \frac{(\lambda_r)^{1/3}}{||\hat{\mathbf{U}}^s_{\cdot r}||_F}\;, \;\;\;\;\hat{\mathbf{V}}^s_{\cdot r} = \hat{\mathbf{V}}^s_{\cdot r} \frac{(\lambda_r)^{1/3}}{||\hat{\mathbf{V}}^s_{\cdot r}||_F}\;, \;\;\; \textnormal{and} \;\hat{\mathbf{W}}^s_{\cdot r} = \hat{\mathbf{W}}^s_{\cdot r} \frac{(\lambda_r)^{1/3}}{||\hat{\mathbf{W}}^s_{\cdot r}||_F}\;,
\end{equation*}
where $\lambda_r = ||\hat{\mathbf{U}}^s_{\cdot r}||_F||\hat{\mathbf{V}}^s_{\cdot r}||_F||\hat{\mathbf{W}}^s_{\cdot r}||_F$, $\hat{\mathbf{U}}^s_{\cdot r}$, $\hat{\mathbf{V}}^s_{\cdot r}$, and $\hat{\mathbf{W}}^s_{\cdot r}$ are the $r$-th column of $\hat{\mathbf{U}}^s$, $\hat{\mathbf{V}}^s$, and $\hat{\mathbf{W}}^s$, respectively. Other parts of Algorithm \ref{alg3} are the same as those presented in Algorithm~1 of the main manuscript. 
In practice, we select the key parameter $R$ using cross-validation, similar to the single-level Algorithm \ref{alg2}.



\clearpage

\section{Additional Simulation Results}\label{sec:supp-sims}
\subsection{Convergence Check for the Single-Level Simulation}
In this subsection, we present additional results for our simulation with single-level functional data. We first present the potential scale reduction factor (PSRF) for the two Bayesian algorithms in Figure \ref{fig:srf_single}. From the results, our BAMIFun algorithm does not have a convergence issue even with $95\%$ missingness and small sample sizes. The BAMITA algorithm, however, has some convergence issues when the sample size is small and the proportion of missingness is high. We therefore increase the number of iterations to $3000$ for the BAMITA algorithm.

\begin{figure}[ht] 
\centering
\includegraphics[width=\textwidth]{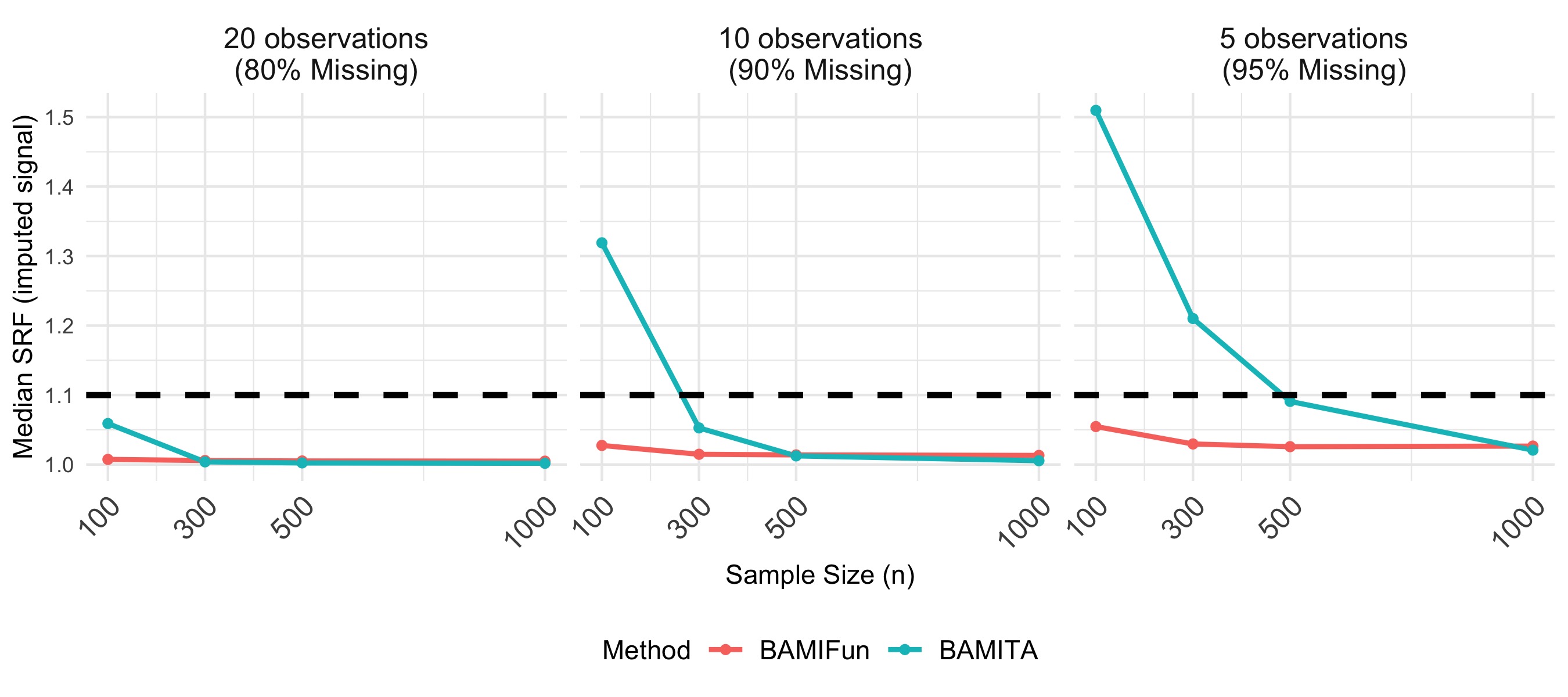}
\caption{Median potential scale reduction factor (PSRF) results across all simulation replications for the single-level functional data. A PSRF close to 1 indicates a successful convergence of the Bayesian posterior samples. We use the PSRF threshold of 1.1 in the figure.}
\label{fig:srf_single}
\end{figure}

\clearpage

\subsection{Simulation With IG Prior on the Smoothing Parameter}
We further conducted a sensitivity analysis on the single-level simulation, replacing the default weakly-informative inverse-gamma $\textnormal{IG}(0.01,0.01)$ prior on $\sigma_B^2$ with a more informative $\textnormal{IG}(2,1)$ prior. As shown in Figure \ref{fig:imputation_sens}, the two priors yield nearly identical results: the relative imputation {\color{black}RMSE} and the element-wise coverage essentially coincide across all sample sizes and missingness levels. The only appreciable difference appears under the most extreme missingness ($95\%$, i.e., five observed time points) at the smallest sample size, where the $\textnormal{IG}(2,1)$ prior is marginally more conservative---a slightly higher {\color{black}RMSE} accompanied by slightly better, closer-to-nominal coverage---and this difference vanishes as the sample size grows. We therefore conclude that the imputation and downstream inference are robust to the choice of prior on $\sigma_B^2$.

\begin{figure}[ht] 
\centering
\includegraphics[width=\textwidth]{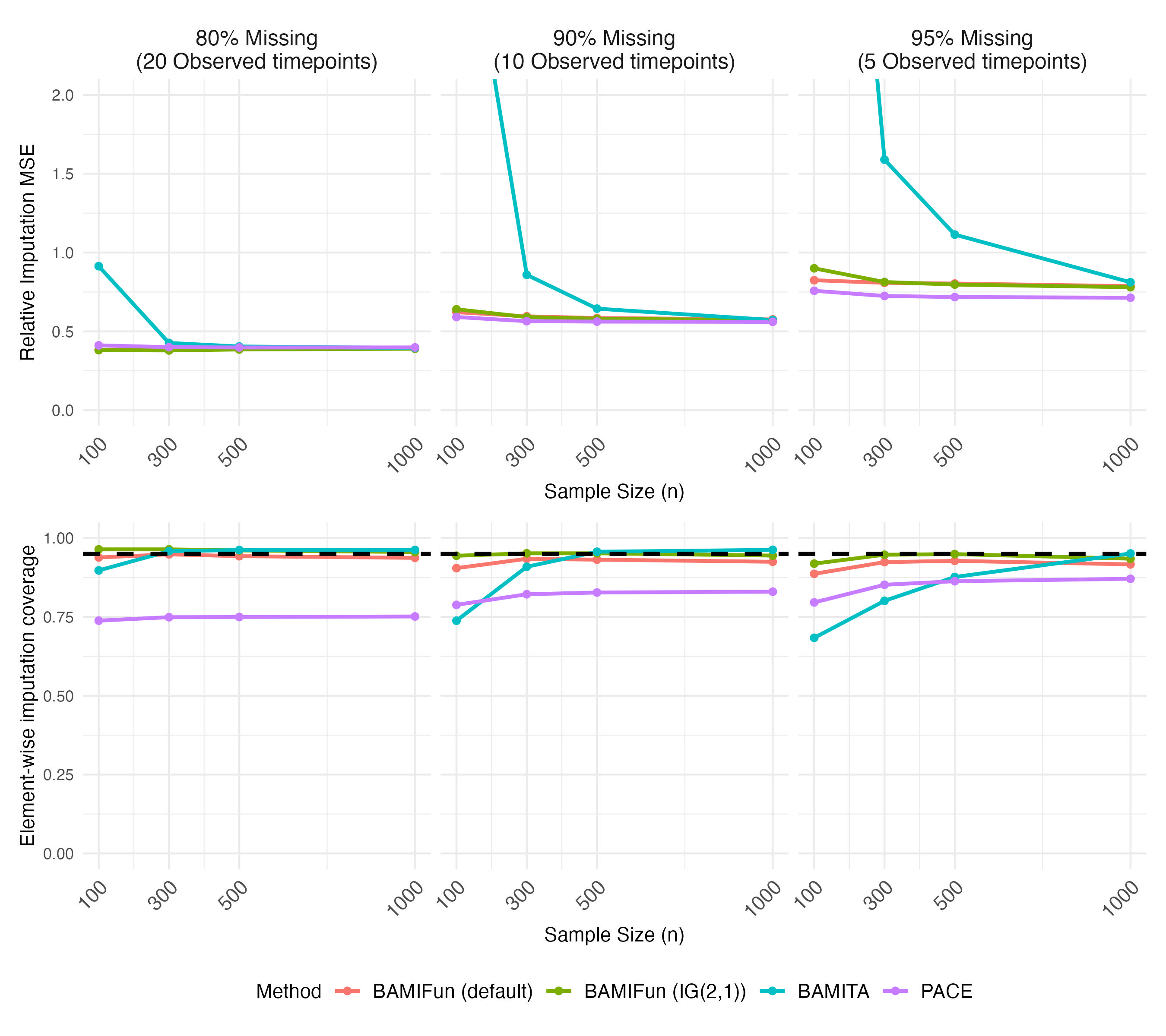}
\caption{Sensitivity of BAMIFun to the prior on the smoothing parameter
$\sigma_B^2$ for single-level functional data. BAMIFun under its default weakly-informative $\textnormal{IG}(0.01,0.01)$ prior on $\sigma_B^2$ is compared with BAMIFun under a more informative $\textnormal{IG}(2,1)$ prior, alongside BAMITA and PACE. }
\label{fig:imputation_sens}
\end{figure}


\subsection{Simulation With an Increasing Number of Time Points}

The relative performance of BAMITA and the smoothness-constraint methods (BAMIFun and PACE) depends critically on how finely the functional domain is sampled. To further demonstrate this point, we conduct an additional single-level experiment that varies the total number of time points $K$ while holding the number of \emph{observed} time points per subject fixed at the average of $10$. Because each subject contributes a fixed amount of information regardless of $K$, refining the grid does not change the intrinsic difficulty of recovering a smooth curve from ${\sim}10$ noisy samples; it only increases the resolution of the target grid, and hence the missing proportion, which is now a derived quantity equal to $1-10/K$.

\subsubsection*{Data-generating mechanism}
We generate single-level functional data using the same mechanism as the first simulation experiment (Section 5.1 of the main manuscript). For each subject $i=1,\ldots,N$,
\begin{equation*}
X_i(t_k)=\sum_{h=1}^{12}\xi_{ih}\,u_h(t_k)+\epsilon_i(t_k),\qquad k=1,\ldots,K,
\end{equation*}
where $\{u_h\}_{h=1}^{12}$ are the $12$ eigenfunctions obtained from the NHANES data application, evaluated on an equally spaced grid of $K$ points (obtained by subsampling the underlying $1440$-point grid); the scores are $\xi_{ih}\sim N(0,\lambda_h)$ with eigenvalues $\{\lambda_h\}_{h=1}^{12}=(2,2,2,1,1,1,0.5,0.5,0.5,\\ 0.1,0.1,0.1)$, and $\epsilon_i(t_k)\stackrel{\text{iid}}{\sim}N(0,1)$. We vary the total number of time points $K\in\{50,100,200,400\}$ (the first simulation experiment used $K=100$) and the sample size $N\in\{100,300\}$. For each subject, the number of observed time points is drawn uniformly on $\{6,\ldots,14\}$ (i.e., $10\pm4$, independent of $K$) at uniformly random locations, and the remaining entries are treated as missing. The resulting missing proportion is therefore $80\%$, $90\%$, $95\%$, and $97.5\%$ for $K=50, 100, 200$, and $400$, respectively. Each scenario is evaluated over $500$ replications.


\subsubsection*{Methods and metrics}
We compare BAMIFun (using $L=15$ B-spline basis functions), the FPCA-based PACE algorithm, and BAMITA, as in the first simulation experiment, except that the rank $R$ is taken directly from the frequentist FPCA fit rather than by cross-validation. We report the relative imputation {\color{black}RMSE} over the missing cells (median over replications) and the element-wise $95\%$ coverage of the imputed values (mean over replications). For PACE, the interval is the point estimate $\pm\,1.96$ times the estimated standard error from \texttt{face.sparse()}; for BAMIFun and BAMITA, it is the $2.5\%$--$97.5\%$ posterior quantile interval.

\subsubsection*{Results}
Figure~\ref{fig:varyingK} summarizes the results. Consistent with the smoothness argument, BAMIFun and PACE---both of which borrow strength across neighboring time points, through a spline basis and estimated eigenfunctions, respectively---remain essentially flat as $K$ increases, with a {\color{black}RMSE} close to $0.6$ across all values of $K$. BAMITA, by contrast, treats time as an unordered categorical mode and cannot pool information across adjacent points; as the grid is refined, its relative imputation {\color{black}RMSE} deteriorates sharply, rising from $0.99$ at $K=50$ to $55.9$ at $K=400$ when $N=100$ (too large to present in the figure), with a corresponding collapse in element-wise coverage from $93.2\%$ to $53.3\%$ at $N=100$. BAMIFun maintains coverage close to the nominal level throughout ($93.3\%$--$94.0\%$ at $N=100$ and $96.8\%$--$98.8\%$ at $N=300$), whereas the single-imputation PACE algorithm, although comparably accurate in {\color{black}RMSE}, severely under-covers (around $79\%$ at $N=100$ and $82\%$ at $N=300$). These results reinforce the central message that incorporating smoothness allows BAMIFun to remain stable and well-calibrated even as the functional grid becomes arbitrarily fine.


\begin{figure}[ht]
\centering
\includegraphics[width=\textwidth]{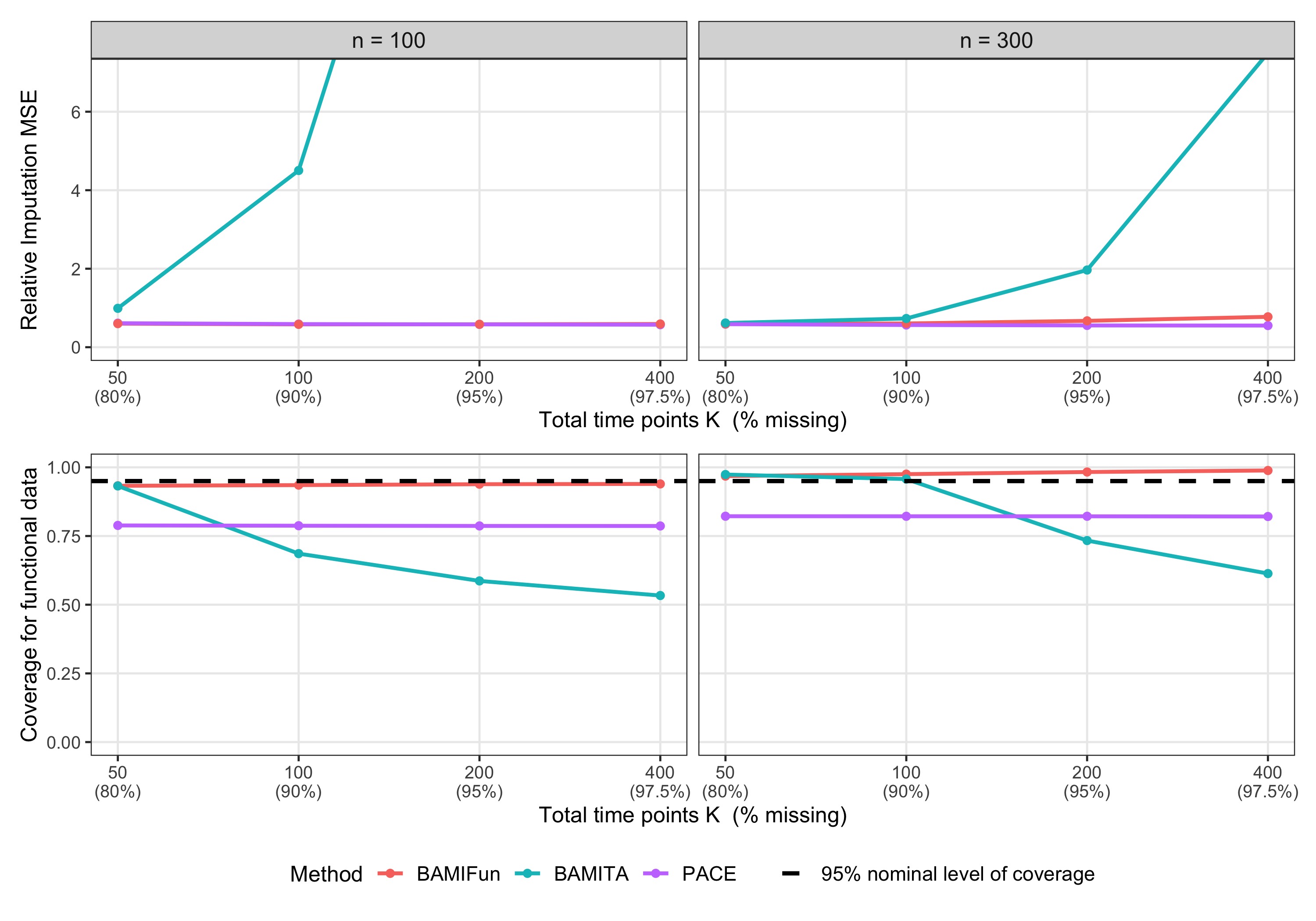}
\caption{Single-level imputation performance as the total number of time points $K$ increases (with ${\sim}10$ observed time points per subject, so that the missing proportion, shown in parentheses on the horizontal axis, equals $1-10/K$). Columns correspond to the sample size $N$; the top row displays the relative imputation {\color{black}RMSE} (truncated at $7$ for visibility) and the bottom row the element-wise $95\%$ coverage, with the dashed line marking the nominal $0.95$ level. As $K$ grows, BAMIFun and PACE remain stable while BAMITA's error grows and its coverage collapses; among the two methods that remain stable in RMSE, only BAMIFun attains near-nominal coverage.}
\label{fig:varyingK}
\end{figure}

\clearpage

\subsection{Convergence Check for the Multiway Simulation}
In this subsection, we present additional results for our simulation with multiway functional data. We first present the potential scale reduction factor (PSRF) for the two Bayesian algorithms in Figure \ref{fig:srf_multi}. Similar to the single-level setting, our BAMIFun algorithm does not have a convergence issue even with $95\%$ missingness and small sample sizes. The BAMITA algorithm, however, exhibits some convergence issues when the sample size is small and the proportion of missingness is high. 

\begin{figure}[ht] 
\centering
\includegraphics[width=\textwidth]{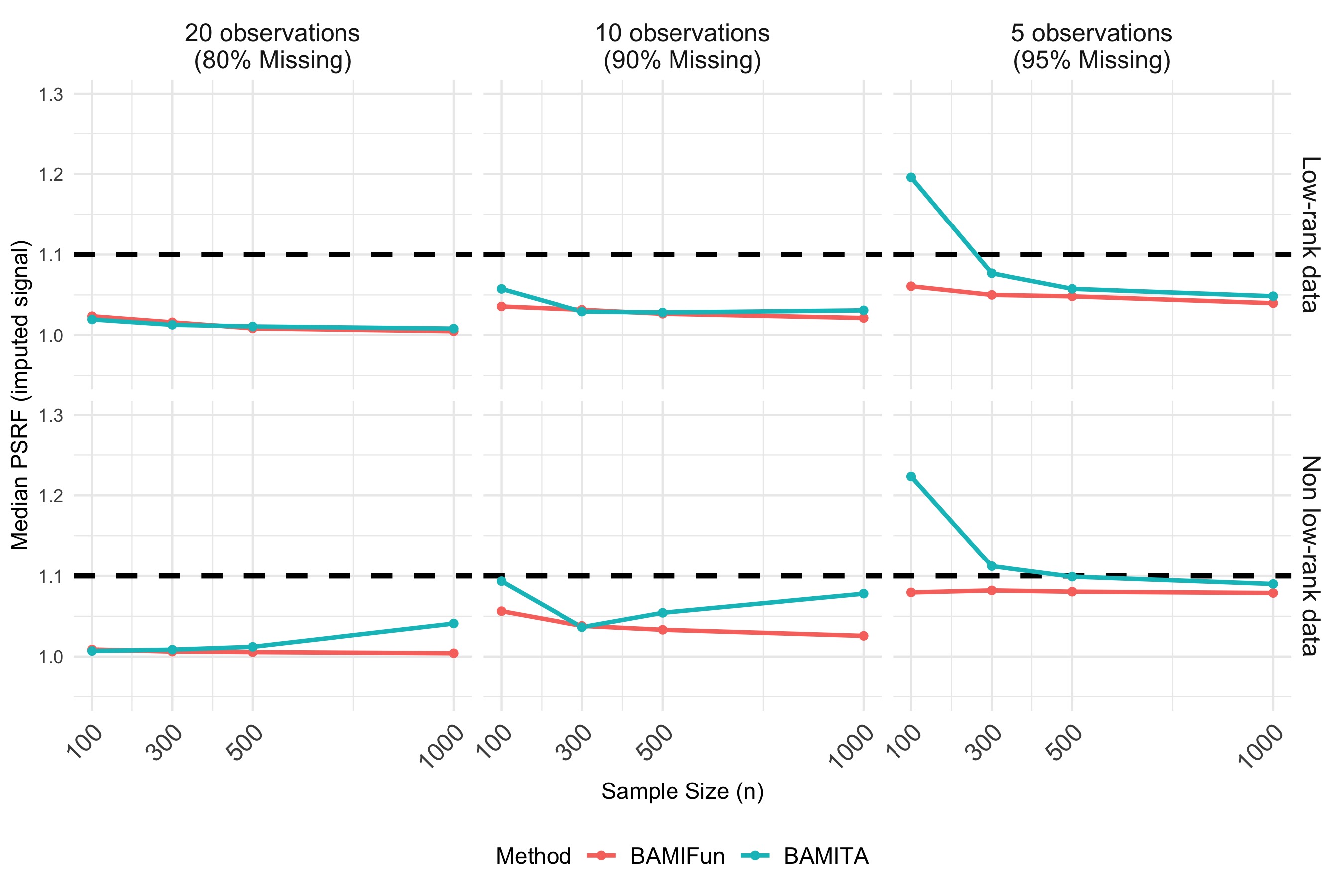}
\caption{Median potential scale reduction factor (PSRF) results across all simulation replications for the multiway functional data. A PSRF close to 1 indicates a successful convergence of the Bayesian posterior samples. We use the PSRF threshold of 1.1 in the figure.}
\label{fig:srf_multi}
\end{figure}

\clearpage

\subsection{Multiway Simulation With Non-Low-Rank Functional Data}
For the scenario without assuming a low rank structure, we generate $X_{i j}(t_k)=\sum_{r=1}^2 a_{i r} \,u_r(t_k)+\sum_{r=1}^4 b_{i j r} \,u_{r+4}(t_k) + \epsilon_{ijk}$ where $a_{i r}, b_{i j r} \sim N(0, \sigma_r^2)$, $\epsilon_{ijk}\sim N(0,1)$ and $u_r$ are again generated according to the data application. Results are presented in Table~\ref{tab:c7v10_nonlowrank}. From the results, our BAMIFun algorithm has a higher {\color{black}RMSE} compared with the PACE algorithm, which is as expected since the true functional data are generated under the multilevel FPCA structure. However, our BAMIFun algorithm still performs better than the BAMITA algorithm with higher stability. 

\begin{table}[ht]
\centering
\small
\setlength{\tabcolsep}{6pt}
\caption{Imputation performance on \textbf{non-low-rank} functional data (effective rank $\approx$ 12): relative imputation {\color{black}RMSE} with the mean 95\% element-wise coverage (in parentheses, \%). Coverage is not reported for PACE.}
\label{tab:c7v10_nonlowrank}
\begin{tabular}{cc ccc}
\toprule
Sample & Observations & \multicolumn{3}{c}{Method} \\
\cmidrule(lr){3-5}
Size & (Missing Prop.) & PACE & BAMIFun & BAMITA \\
\midrule
100 & 20 (80\%) & 0.258($-$) & 0.484(94.3) & 0.515(94.2) \\
 & 10 (90\%) & 0.325($-$) & 0.524(94.1) & 0.649(93.3) \\
 & 5 (95\%) & 0.437($-$) & 0.597(93.4) & 1.390(89.0) \\
\addlinespace
300 & 20 (80\%) & 0.251($-$) & 0.495(94.3) & 0.504(94.3) \\
 & 10 (90\%) & 0.306($-$) & 0.528(94.1) & 0.553(93.9) \\
 & 5 (95\%) & 0.403($-$) & 0.590(93.6) & 0.677(92.8) \\
\addlinespace
500 & 20 (80\%) & 0.249($-$) & 0.496(94.3) & 0.502(94.3) \\
 & 10 (90\%) & 0.302($-$) & 0.528(94.1) & 0.545(94.0) \\
 & 5 (95\%) & 0.394($-$) & 0.588(93.7) & 0.632(93.2) \\
\addlinespace
1000 & 20 (80\%) & 0.247($-$) & 0.499(94.3) & 0.501(94.3) \\
 & 10 (90\%) & 0.298($-$) & 0.532(94.1) & 0.539(94.1) \\
 & 5 (95\%) & 0.387($-$) & 0.588(93.8) & 0.606(93.5) \\
\bottomrule
\end{tabular}
\end{table}


\subsection{Computational Cost}
Figures~\ref{fig:time_single} and \ref{fig:time_multi} report the median per-replication runtime of each method in the simulation study. The frequentist FPCA imputation (PACE/MFPCA) is fast, completing in under five seconds per replication across all configurations.
BAMIFun is more expensive: in the single-level simulation it ranges from roughly $40$~seconds at $N=100$ to about $7$~minutes at $N=1{,}000$, and in the multiway simulation from roughly $2$~minutes at $N=100$ to about $18$~minutes at $N=1{,}000$, reflecting the additional cost of its penalized-spline smoothness updates within the Gibbs sampler. For all methods the runtime scales approximately linearly in the sample size $N$
and is essentially flat across the missing-data levels---the per-iteration cost is governed by the size of the data grid and the number of latent components rather than by the missing fraction---and it is comparable across the low-rank and non-low-rank data-generating scenarios.

\begin{figure}[ht] 
\centering
\includegraphics[width=\textwidth]{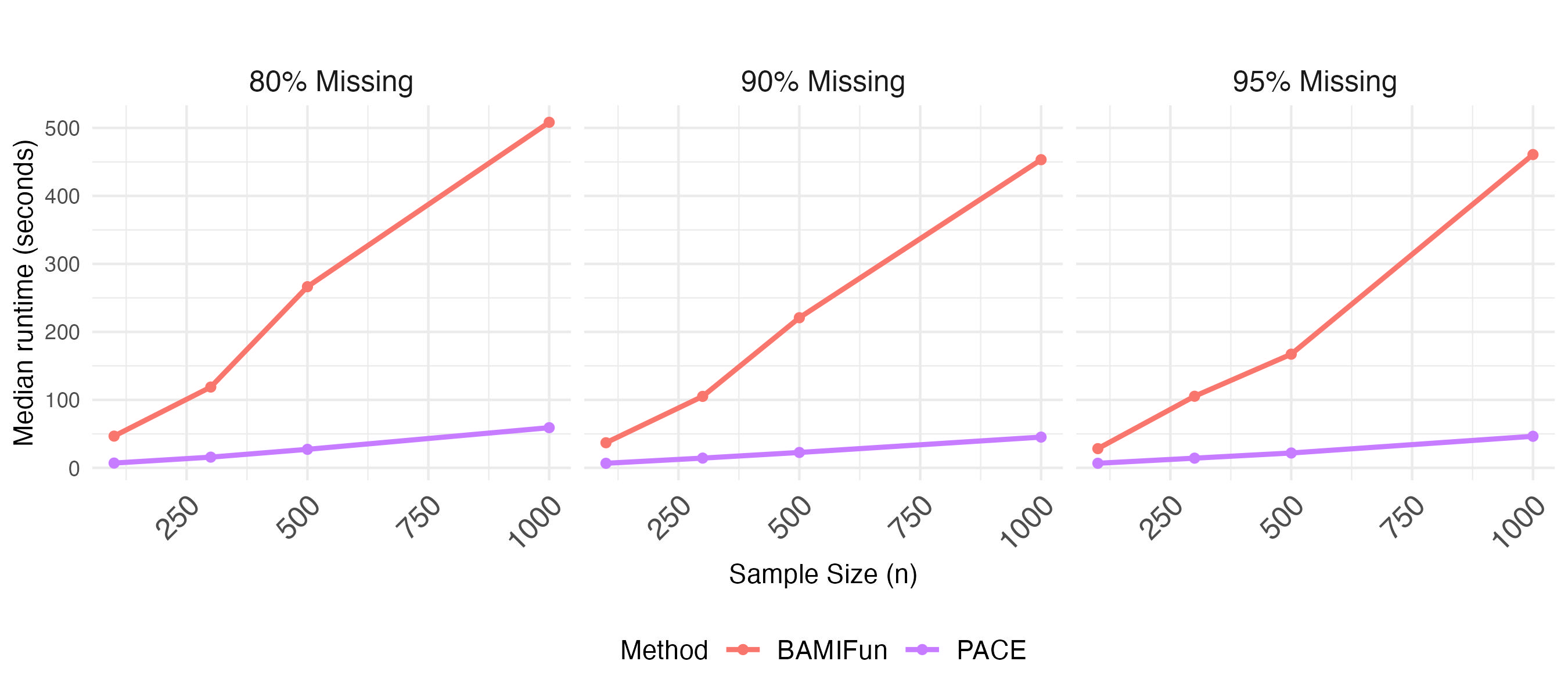}
\caption{Median per-replication runtime for the BAMIFun, PACE, and BAMITA algorithms in the single-level functional data simulation across sample sizes.}
\label{fig:time_single}
\end{figure}

\begin{figure}[ht] 
\centering
\includegraphics[width=\textwidth]{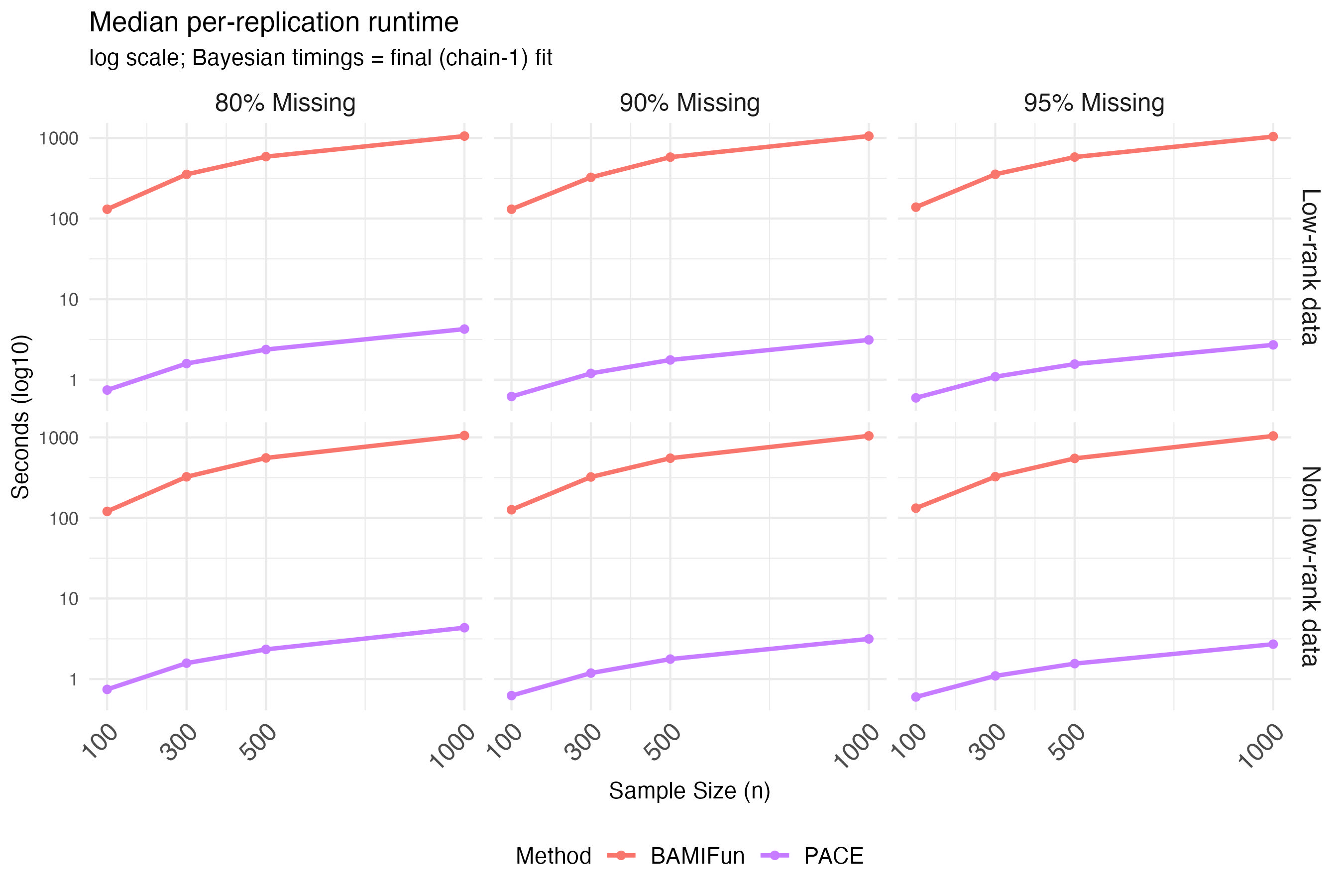}
\caption{Median per-replication runtime for the BAMIFun, PACE, and BAMITA algorithms in the multiway functional data simulation across sample sizes.}
\label{fig:time_multi}
\end{figure}

\section{Additional Data Application Results}\label{sec:supp-apps}
\subsection{NHANES Data Downstream Analysis}
To demonstrate that BAMIFun carries through to a downstream inferential task, we revisit the NHANES accelerometry data and relate the subject-specific physical-activity profile to all-cause mortality through a scalar-on-function regression. We restrict attention to adults aged at least $50$ with complete covariate information and analyze a random sample of $n = 3{,}880$ subjects. For subject $i$, let $Y_i \in \{0,1\}$ denote the mortality indicator, $X_i(t)$ the Monitor-Independent Movement Summary (MIMS) activity intensity at time of day $t$ over the $24$-hour day $\mathcal{T}$, and $\mathbf{Z}_i = (\text{age}_i,\text{gender}_i,\text{BMI}_i)^{\top}$ a vector of scalar covariates. The downstream model is the logistic scalar-on-function regression
\begin{equation}
  \operatorname{logit}\bigl\{\Pr(Y_i = 1)\bigr\}
  \;=\; \alpha \;+\; \int_{\mathcal{T}} \beta(t)\,X_i(t)\,\mathrm{d}t
  \;+\; \mathbf{Z}_i^{\top}\boldsymbol{\gamma},
  \qquad i = 1,\dots,n,
  \label{eq:sofr}
\end{equation}
where $\beta(t)$ is the functional coefficient, the log-odds of mortality per unit of activity at time of day $t$. We represent $\beta(\cdot)$ with a cubic penalized regression spline and estimate \eqref{eq:sofr} by penalized likelihood with smoothing parameter selected by REML. 

To create the sparsely sampled regime that BAMIFun targets, and to enable a comparison against a known complete-data benchmark, we induce missingness completely at random on the fully observed activity profiles at three levels ($90\%$, $95\%$, and $97.5\%$ missing). We then generate the imputed activity profiles from the BAMIFun posterior, with the number of components selected by functional principal component analysis, and refit \eqref{eq:sofr} using each imputed dataset. Denote by $\widehat{\beta}^{(m)}(t)$ the estimate from the $m$-th imputation, we propagate the imputation uncertainty into the downstream coefficient by Rubin's rules. 

Figure~\ref{fig:nhanes-downstream} summarizes the downstream analysis results. The estimated functional coefficient $\widehat{\beta}(t)$ traces a clear contrast: positive in the early-morning and overnight hours and negative around midday. Crucially, the BAMIFun multiple-imputation estimate recovers this complete-data association at every missing-data level: the pooled $\bar{\beta}(t)$ closely tracks the oracle coefficient fit to the unmasked profiles even when $90$--$97.5\%$ of each subject's curve is unobserved. The inference, however, correctly reflects the information lost to missingness. Measured against the complete-data coefficient, the Rubin band attains nominal coverage at $90\%$ and $95\%$ missing (containing the oracle $\widehat{\beta}(t)$ at all evaluated time points) and undercovers only at the most extreme $97.5\%$ setting.

\begin{figure}[ht] 
\centering
\includegraphics[width=\textwidth]{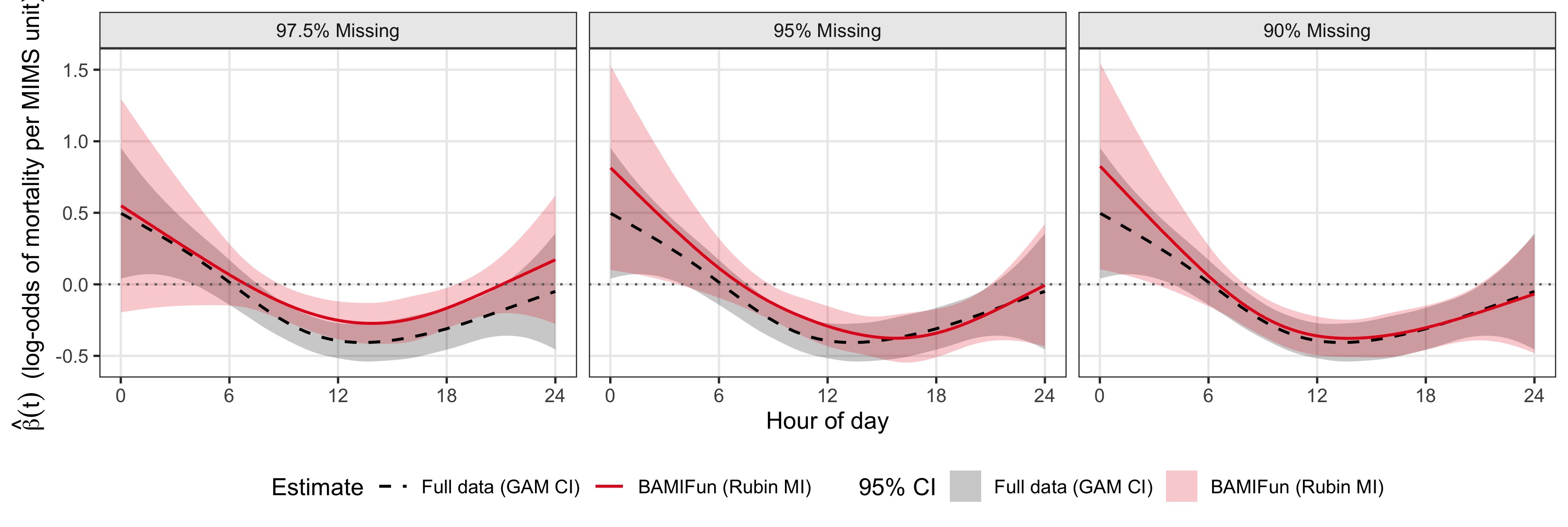}
\caption{BAMIFun recovers the daily activity–mortality association with honestly propagated imputation uncertainty (NHANES)}
\label{fig:nhanes-downstream}
\end{figure}

\clearpage


\begin{thebibliography}{}

\bibitem[\protect\citeauthoryear{Chiou, Chen, and Yang}{Chiou
  et~al.}{2014}]{chiou2014multivariate}
Chiou, J.-M., Chen, Y.-T., and Yang, Y.-F. (2014).
\newblock Multivariate functional principal component analysis: A normalization
  approach.
\newblock {\em Statistica Sinica} {\bf 24,} 1571--1596.

\bibitem[\protect\citeauthoryear{Cong, Judge, Xu, Diallo, Janton, Brownell,
  Maas, and Graf}{Cong et~al.}{2017}]{cong2017influence}
Cong, X., Judge, M., Xu, W., Diallo, A., Janton, S., Brownell, E.~A., Maas, K.,
  and Graf, J. (2017).
\newblock Influence of feeding type on gut microbiome development in
  hospitalized preterm infants.
\newblock {\em Nursing Research} {\bf 66,} 123--133.

\bibitem[\protect\citeauthoryear{Crainiceanu, Goldsmith, Leroux, and
  Cui}{Crainiceanu et~al.}{2024}]{crainiceanu2024functional}
Crainiceanu, C.~M., Goldsmith, J., Leroux, A., and Cui, E. (2024).
\newblock {\em Functional Data Analysis with R}.
\newblock Chapman and Hall/CRC.

\bibitem[\protect\citeauthoryear{Cui, Crainiceanu, and Leroux}{Cui
  et~al.}{2021}]{cui2021additive}
Cui, E., Crainiceanu, C.~M., and Leroux, A. (2021).
\newblock Additive functional {C}ox model.
\newblock {\em Journal of Computational and Graphical Statistics} {\bf 30,}
  780--793.

\bibitem[\protect\citeauthoryear{Cui, Li, Crainiceanu, and Xiao}{Cui
  et~al.}{2023}]{cui2023fast}
Cui, E., Li, R., Crainiceanu, C.~M., and Xiao, L. (2023).
\newblock Fast multilevel functional principal component analysis.
\newblock {\em Journal of Computational and Graphical Statistics} {\bf 32,}
  366--377.

\bibitem[\protect\citeauthoryear{Cui, Thompson, Carroll, and Ruppert}{Cui
  et~al.}{2022}]{cui2022semiparametric}
Cui, E., Thompson, E.~C., Carroll, R.~J., and Ruppert, D. (2022).
\newblock A semiparametric risk score for physical activity.
\newblock {\em Statistics in Medicine} {\bf 41,} 1191--1204.

\bibitem[\protect\citeauthoryear{Di, Crainiceanu, Caffo, and Punjabi}{Di
  et~al.}{2009}]{di2009multilevel}
Di, C., Crainiceanu, C.~M., Caffo, B.~S., and Punjabi, N.~M. (2009).
\newblock Multilevel functional principal component analysis.
\newblock {\em The Annals of Applied Statistics} {\bf 3,} 458--488.

\bibitem[\protect\citeauthoryear{Gelman and Rubin}{Gelman and
  Rubin}{1992}]{gelman1992inference}
Gelman, A. and Rubin, D.~B. (1992).
\newblock Inference from iterative simulation using multiple sequences.
\newblock {\em Statistical Science} {\bf 7,} 457--472.

\bibitem[\protect\citeauthoryear{Greven, Crainiceanu, Caffo, and Reich}{Greven
  et~al.}{2010}]{greven2011longitudinal}
Greven, S., Crainiceanu, C., Caffo, B., and Reich, D. (2010).
\newblock Longitudinal functional principal component analysis.
\newblock {\em Electronic Journal of Statistics} {\bf 4,} 1022--1054.

\bibitem[\protect\citeauthoryear{Han, Shi, and Zhang}{Han
  et~al.}{2024}]{han2024guaranteed}
Han, R., Shi, P., and Zhang, A.~R. (2024).
\newblock Guaranteed functional tensor singular value decomposition.
\newblock {\em Journal of the American Statistical Association} {\bf 119,}
  995--1007.

\bibitem[\protect\citeauthoryear{Happ and Greven}{Happ and
  Greven}{2018}]{happ2018multivariate}
Happ, C. and Greven, S. (2018).
\newblock Multivariate functional principal component analysis for data
  observed on different (dimensional) domains.
\newblock {\em Journal of the American Statistical Association} {\bf 113,}
  649--659.

\bibitem[\protect\citeauthoryear{He, Yucel, and Raghunathan}{He
  et~al.}{2011}]{he2011functional}
He, Y., Yucel, R., and Raghunathan, T.~E. (2011).
\newblock A functional multiple imputation approach to incomplete longitudinal
  data.
\newblock {\em Statistics in Medicine} {\bf 30,} 1137--1156.

\bibitem[\protect\citeauthoryear{Jang, Manatunga, Chang, and Long}{Jang
  et~al.}{2021}]{jang2021bayesian}
Jang, J.~H., Manatunga, A.~K., Chang, C., and Long, Q. (2021).
\newblock A bayesian multiple imputation approach to bivariate functional data
  with missing components.
\newblock {\em Statistics in Medicine} {\bf 40,} 4772--4793.

\bibitem[\protect\citeauthoryear{Jiang, Crainiceanu, and Cui}{Jiang
  et~al.}{2025}]{jiang2025tutorial}
Jiang, Z., Crainiceanu, C., and Cui, E. (2025).
\newblock Tutorial on bayesian functional regression using stan.
\newblock {\em Statistics in Medicine} {\bf 44,} e70265.

\bibitem[\protect\citeauthoryear{Jiang, Li, and Lock}{Jiang
  et~al.}{2025}]{jiang2025bamita}
Jiang, Z., Li, G., and Lock, E.~F. (2025).
\newblock {BAMITA}: Bayesian multiple imputation for tensor arrays.
\newblock {\em Biostatistics} {\bf 26,} kxae047.

\bibitem[\protect\citeauthoryear{John, Tang, Albinali, and Intille}{John
  et~al.}{2019}]{john2019open}
John, D., Tang, Q., Albinali, F., and Intille, S. (2019).
\newblock An open-source monitor-independent movement summary for accelerometer
  data processing.
\newblock {\em Journal for the Measurement of Physical Behaviour} {\bf 2,}
  268--281.

\bibitem[\protect\citeauthoryear{Kidzi{\'n}ski and Hastie}{Kidzi{\'n}ski and
  Hastie}{2024}]{kidzinski2024modeling}
Kidzi{\'n}ski, {\L}. and Hastie, T. (2024).
\newblock Modeling longitudinal data using matrix completion.
\newblock {\em Journal of Computational and Graphical Statistics} {\bf 33,}
  551--566.

\bibitem[\protect\citeauthoryear{Kong, Staicu, and Maity}{Kong
  et~al.}{2016}]{kong2016classical}
Kong, D., Staicu, A.-M., and Maity, A. (2016).
\newblock Classical testing in functional linear models.
\newblock {\em Journal of Nonparametric Statistics} {\bf 28,} 813--838.

\bibitem[\protect\citeauthoryear{Lang and Brezger}{Lang and
  Brezger}{2004}]{lang2004bayesian}
Lang, S. and Brezger, A. (2004).
\newblock Bayesian p-splines.
\newblock {\em Journal of Computational and Graphical Statistics} {\bf 13,}
  183--212.

\bibitem[\protect\citeauthoryear{Li and Lock}{Li and
  Lock}{2026}]{li2025integrative}
Li, G. and Lock, E.~F. (2026).
\newblock Integrative analysis of multimodal omics data.
\newblock {\em Annual Review of Statistics and Its Application} {\bf 13,}
  347--367.

\bibitem[\protect\citeauthoryear{Matuk, Bharath, Chkrebtii, and Kurtek}{Matuk
  et~al.}{2022}]{matuk2022bayesian}
Matuk, J., Bharath, K., Chkrebtii, O., and Kurtek, S. (2022).
\newblock Bayesian framework for simultaneous registration and estimation of
  noisy, sparse, and fragmented functional data.
\newblock {\em Journal of the American Statistical Association} {\bf 117,}
  1964--1980.

\bibitem[\protect\citeauthoryear{Muirhead}{Muirhead}{2009}]{muirhead2009aspects}
Muirhead, R.~J. (2009).
\newblock {\em Aspects of multivariate statistical theory}.
\newblock John Wiley \& Sons.

\bibitem[\protect\citeauthoryear{Petrovich, Reimherr, and Daymont}{Petrovich
  et~al.}{2022}]{petrovich2022highly}
Petrovich, J., Reimherr, M., and Daymont, C. (2022).
\newblock Highly irregular functional generalized linear regression with
  electronic health records.
\newblock {\em Journal of the Royal Statistical Society Series C: Applied
  Statistics} {\bf 71,} 806--833.

\bibitem[\protect\citeauthoryear{Rao and Reimherr}{Rao and
  Reimherr}{2021}]{rao2021modern}
Rao, A.~R. and Reimherr, M. (2021).
\newblock Modern multiple imputation with functional data.
\newblock {\em Stat} {\bf 10,} e331.

\bibitem[\protect\citeauthoryear{Rubin}{Rubin}{1996}]{rubin1996multiple}
Rubin, D.~B. (1996).
\newblock Multiple imputation after 18+ years.
\newblock {\em Journal of the American Statistical Association} {\bf 91,}
  473--489.

\bibitem[\protect\citeauthoryear{Ruppert}{Ruppert}{2002}]{ruppert2002selecting}
Ruppert, D. (2002).
\newblock Selecting the number of knots for penalized splines.
\newblock {\em Journal of Computational and Graphical Statistics} {\bf 11,}
  735--757.

\bibitem[\protect\citeauthoryear{Ruppert, Wand, and Carroll}{Ruppert
  et~al.}{2003}]{ruppert2003semiparametric}
Ruppert, D., Wand, M.~P., and Carroll, R.~J. (2003).
\newblock {\em Semiparametric regression}.
\newblock Number~12. Cambridge university press.

\bibitem[\protect\citeauthoryear{Shou, Zipunnikov, Crainiceanu, and
  Greven}{Shou et~al.}{2015}]{shou2015structured}
Shou, H., Zipunnikov, V., Crainiceanu, C.~M., and Greven, S. (2015).
\newblock Structured functional principal component analysis.
\newblock {\em Biometrics} {\bf 71,} 247--257.

\bibitem[\protect\citeauthoryear{Speckman and Sun}{Speckman and
  Sun}{2003}]{speckman2003fully}
Speckman, P.~L. and Sun, D. (2003).
\newblock Fully bayesian spline smoothing and intrinsic autoregressive priors.
\newblock {\em Biometrika} pages 289--302.

\bibitem[\protect\citeauthoryear{Sun, Tsutakawa, and He}{Sun
  et~al.}{2001}]{sun2001propriety}
Sun, D., Tsutakawa, R.~K., and He, Z. (2001).
\newblock Propriety of posteriors with improper priors in hierarchical linear
  mixed models.
\newblock {\em Statistica Sinica} pages 77--95.

\bibitem[\protect\citeauthoryear{Wood and Wood}{Wood and
  Wood}{2015}]{wood2015package}
Wood, S. and Wood, M.~S. (2015).
\newblock Package ‘mgcv’.
\newblock {\em R package version} {\bf 1,} 729.

\bibitem[\protect\citeauthoryear{Xiao, Li, Checkley, and Crainiceanu}{Xiao
  et~al.}{2018}]{xiao2018fast}
Xiao, L., Li, C., Checkley, W., and Crainiceanu, C. (2018).
\newblock Fast covariance estimation for sparse functional data.
\newblock {\em Statistics and Computing} {\bf 28,} 511--522.

\bibitem[\protect\citeauthoryear{Yao, M{\"u}ller, and Wang}{Yao
  et~al.}{2005}]{yao2005functional}
Yao, F., M{\"u}ller, H.-G., and Wang, J.-L. (2005).
\newblock Functional data analysis for sparse longitudinal data.
\newblock {\em Journal of the American Statistical Association} {\bf 100,}
  577--590.

\bibitem[\protect\citeauthoryear{Zipunnikov, Caffo, Yousem, Davatzikos,
  Schwartz, and Crainiceanu}{Zipunnikov
  et~al.}{2011}]{zipunnikov2011multilevel}
Zipunnikov, V., Caffo, B., Yousem, D.~M., Davatzikos, C., Schwartz, B.~S., and
  Crainiceanu, C. (2011).
\newblock Multilevel functional principal component analysis for
  high-dimensional data.
\newblock {\em Journal of Computational and Graphical Statistics} {\bf 20,}
  852--873.

\end{thebibliography}
\end{document}